\documentclass[runningheads]{llncs}
\usepackage[T1]{fontenc}
\usepackage{graphicx}
\usepackage{comment}
\usepackage{makeidx}
\usepackage{amsmath}
\usepackage{amsfonts}
\usepackage{amssymb}
\usepackage{xcolor}
\usepackage{gastex}
\usepackage{rotating}
\usepackage{slashbox}
\usepackage[misc,geometry]{ifsym}
\usepackage{geometry}
\usepackage[colorlinks=true]{hyperref}
\usepackage{multicol}
\usepackage{tikz}
\usetikzlibrary{trees}
\usetikzlibrary{arrows}
\usetikzlibrary{calc}
\usepackage[ruled,boxruled]{algorithm2e}
\usepackage{venndiagram}
\usepackage{adjustbox}
\begin{document}
\title{The expressive power of revised Datalog on problems with closure properties}
%
%
\author{Shiguang Feng\orcidID{0000-0002-5110-3881}}
\authorrunning{S. Feng}
%
\institute{School of Computer Science and Engineering, Sun Yat-sen University, Guangzhou, 510006, China\\
\email{fengshg3@mail.sysu.edu.cn}}
\maketitle              
\begin{abstract}
In this paper, we study the expressive power of revised Datalog on the problems that are closed under substructures.
We show that revised Datalog cannot define all the problems that are in PTIME and closed under substructures. As a corollary, LFP cannot define all the extension-closed problems that are in PTIME.

\end{abstract}
\section{Introduction}
Datalog and its variants are widely used in artificial intelligence and other fields, such as deductive database, knowledge representation, data integration, cloud computing, etc~\cite{doan2012dataintergration,alvaro2003database,gurevich2012datalog,marinescu2018cloud,Revesz1998constraint,siekmann2014historylogic}. 
As a declarative programming language, it is often used to perform data analysis and create complex queries. The complexity and expressive power is an important issue of the study~\cite{abiteboul1995foundations,dantsin2001complexity,Revesz1998constraint,Schlipf1995complexity,Shmueli1987decidability}.
With the recursive computing ability, Datalog is more powerful than first-order logic. It defines exactly the polynomial time computable queries on ordered finite structures~\cite{ebbinghaus1995}. Hence, Datalog captures the complexity class PTIME on ordered finite structures. While on all finite structures, the expressive power of Datalog is very limited. It even cannot define the parity of a set~\cite{ebbinghaus1995}. A Datalog program is constituted of a set of Horn clauses. The characteristics of syntax determine the monotonicity properties of its semantics. That is, every Datalog (resp., positive Datalog, the fragment of Datalog where no negated atomic formula occurs in the body of any clauses) definable query is preserved under extensions~\cite{afrati1995datalog} (resp., homomorphisms~\cite{ajtai1994datalog}). It is natural to ask from the point of view of descriptive complexity that whether Datalog (resp., positive Datalog) captures the polynomial time computable problems that are closed under extensions (resp., homomorphisms). The answer is negative by the work of Afrati et al. who showed that positive Datalog cannot express all monotone queries computable in polynomial time, and the perfect squares problem that is in polynomial time and closed under extensions is not expressible in Datalog~\cite{afrati1995datalog}.

In model theory, many preservation theorems are proved to show the relationship between the closure properties and the syntactic properties of formulas. 
Most of these preservation theorems fail when restricted to finite structures. A lot of research about the preservation theorems on Datalog, first-order logic (FO) and least fixpoint logic (LFP) have been conducted on finite structures. Ajtai and Gurevich showed that a positive Datalog formula is bounded iff it is definable in positive existential first-order logic, and every first-order logic expressible positive Datalog formula is bounded~\cite{ajtai1994datalog}, where a Datalog formula is bounded if there exists a number $n$ such that the fixpoint of the formula can be reached for any finite structure within $n$ steps.
Dawar and Kreutzer showed that the homomorphism preservation theorem fails for LFP, both in general and in restriction to finite structures~\cite{dawar2008datalog}. That is, there is an LFP formula that is preserved under homomorphisms (in the finite) but is not equivalent (in the finite) to a Datalog formula.
The paper~\cite{ketsman2020datalog} studied Datalog with negation and monotonicity, and the expressive power with respect to monotone and homomorphism properties.
The papers~\cite{Dawar2021extension,Rosen1995finite} studied the preservation results under extensions for FO and Datalog. All the results are summarized in Fig.~\ref{fig-summary}.

\begin{figure}
	\begin{center}
		\begin{adjustbox}{scale = 0.8}
			\begin{tikzpicture}[font=\footnotesize]
				\draw[rounded corners] (-0.7,-1.2) rectangle (2.85,-1.8);
				\draw (1.1,-1.5) node [] {bounded pos-Datalog};
				\draw[blue, thick, dash pattern=on 8pt off 3pt, latex-latex] (0.3,0.1) -- (0.3,-1.15); 
				\draw[red, thick, dash pattern=on 8pt off 3pt, latex-latex] (1.55,0.1) -- (1.55,-1.15); 
				\draw[red, thick, dash pattern=on 8pt off 3pt, latex-latex] (2.9,-1.5) -- (4.15,-1.5); 
				\draw[rounded corners] (-0.4,0.15) rectangle (2.35,0.75);
				\draw (1,0.45) node [] {pos-Datalog[FO]};
				\draw[blue, thick, -latex] (0.3,0.85) -- (0.3,2.15); 
				\draw[red, thick, -latex] (1.55,0.85) -- (1.55,2.15); 
				\draw[red, thick, -latex] (-0.45,0.5) -- (-1.95,1.4); 
				\draw[blue, thick, -latex] (-0.45,0.35) -- (-1.95,1.25); 
				\draw[rounded corners] (-0.5,2.2) rectangle (2.45,2.8);
				\draw (1,2.5) node [] {bounded Datalog};
				\draw[blue, thick, -latex] (0.3,2.85) -- (0.3,4.15); 
				\draw[red, thick, -latex] (1.55,2.85) -- (1.55,4.15); 
				\draw[blue, thick, -latex] (2.1,2.85) -- (2.1,6.15); 
				\draw[red, thick, dash pattern=on 8pt off 3pt, latex-latex] (2.55,2.5) -- (6.65,2.5); 
				\draw[rounded corners] (-0.3,4.2) rectangle (1.85,4.8);
				\draw (0.8,4.5) node [] {Datalog[FO]};
				\draw[blue, thick, -latex] (0.3,4.85) -- (0.3,8.15); 
				\draw[red, thick, -latex] (1.55,4.85) -- (1.55,6.15); 
				\draw[red, thick, -latex] (-0.35,4.5) -- (-2.15,5.4); 
				\draw[blue, thick, -latex] (-0.35,4.35) -- (-2.15,5.25); 
				\draw[black, very thick, loosely dotted, -latex] (1.95,4.5) -- (6.6,4.5); 
				\draw[rounded corners] (0.5,6.2) rectangle (3.5,6.8);
				\draw (2,6.5) node [] {bounded Datalog$^r$};
				\draw[blue, thick, -latex] (2.1,6.85) -- (2.1,8.15); 
				\draw[red, thick, dash pattern=on 8pt off 3pt, latex-latex] (1.55,6.85) -- (1.55,8.15); 
				\draw[rounded corners] (-0.2,8.2) rectangle (2.15,8.8);
				\draw (1,8.5) node [] {Datalog$^r$[FO]};
				\draw[red, thick, -latex] (-0.25,8.6) -- (-2.5,9.65); 
				\draw[blue, thick, -latex] (-0.25,8.45) -- (-2.85,9.65); 

				\draw[rounded corners] (7.85,4.2) rectangle (6.65,4.8);
				\draw (7.25,4.5) node [] {FO[E]};
				\draw[red, thick, -latex] (7.3,4.85) -- (7.7,5.25); 
				\draw[red, thick, -latex] (7.1,4.85) -- (7.1,8.15); 
				\draw[rounded corners] (7.9,-1.2) rectangle (6.7,-1.8);
				\draw (7.3,-1.5) node [] {FO[H]};
				\draw[red, thick, dash pattern=on 8pt off 3pt, latex-latex] (5.75,-1.5) -- (6.65,-1.5); 
				\draw[red, thick, -latex] (7.1,-1.15) -- (7.1,2.15); 
				
				\draw[rounded corners] (4.2,-1.2) rectangle (5.7,-1.8);
				\draw (5,-1.5) node [] {pos-$\exists$FO};
				\draw[red, thick, -latex] (5,-1.15) -- (5,1.15); 
				\draw[rounded corners] (6.7,2.2) rectangle (7.7,2.8);
				\draw (7.2,2.5) node [] {$\exists$FO};
				\draw[red, thick, -latex] (7.1,2.85) -- (7.1,4.15); 
				
				\draw[rounded corners] (6.6,8.2) rectangle (7.6,8.8);
				\draw (7.1,8.5) node [] {FO};
				\draw[red, thick, dash pattern=on 8pt off 3pt, latex-latex] (2.2,8.45) -- (6.55,8.45); 
				\draw[red, thick, -latex] (7.1,8.85) -- (5.85,9.8); 
				\draw[rounded corners] (4.1,9.7) rectangle (5.8,10.3);
				\draw (5,10) node [] {LFP};
				\draw[red, thick, -latex] (5.85,10) -- (8.55,10); 
				
				\draw[rounded corners] (-2,1.2) rectangle (-4,1.8);
				\draw (-3,1.5) node [] {pos-Datalog};
				\draw[blue, thick, -latex] (-3.1,1.85) -- (-3.1,3.15); 
				\draw[red, thick, dash pattern=on 8pt off 3pt, latex-latex] (-2.9,1.85) -- (-2.9,3.15); 
				\draw[rounded corners] (-2.1,3.2) rectangle (-3.9,3.8);
				\draw (-3,3.5) node [] {Datalog[H]};
				\draw[blue, thick, -latex] (-3.1,3.85) -- (-3.1,5.15); 
				\draw[red, thick, -latex] (-2.9,3.85) -- (-2.9,5.15); 
				\draw[rounded corners] (-2.2,5.2) rectangle (-3.8,5.8);
				\draw (-3,5.5) node [] {Datalog};
				\draw[blue, thick, -latex] (-3.1,5.85) -- (-3.1,9.65); 
				\draw[red, thick, -latex] (-2.9,5.85) -- (-2.9,9.65); 
				\draw[rounded corners] (-2.1,9.7) rectangle (-3.9,10.3);
				\draw (-3,10) node [] {Datalog$^r$};
				\draw[red, thick, dash pattern=on 8pt off 3pt, latex-latex] (-2.05,10) -- (4.05,10); 
				
				\draw[rounded corners] (4.5,5.2) rectangle (5.5,5.8);
				\draw (5,5.5) node [] {$\exists$LFP};
				\draw[red, thick, -latex] (5,5.85) -- (5,9.65); 
				\draw[red, thick, dash pattern=on 8pt off 3pt, latex-latex] (-2.1,5.5) -- (4.45,5.5); 
				\draw[black, very thick, loosely dotted, -latex] (5.8,5.5) -- (7.65,5.5); 
				
				\draw[rounded corners] (4.15,1.2) rectangle (5.8,1.8);
				\draw (5,1.5) node [] {pos-$\exists$LFP};
				\draw[red, thick, -latex] (5,1.85) -- (5,5.15); 
				\draw[red, thick, dash pattern=on 8pt off 3pt, latex-latex] (-1.95,1.5) -- (4.1,1.5); 

				\draw[rounded corners] (7.7,5.2) rectangle (9.1,5.8);
				\draw (8.4,5.5) node [] {LFP[E]};
				\draw[red, thick, -latex] (8.45,5.85) -- (8.45,7.15); 
				\draw[rounded corners] (8.2,1.2) rectangle (9.6,1.8);
				\draw (8.9,1.5) node [] {LFP[H]};
				\draw[red, thick, -latex] (5.9,1.5) -- (8.15,1.5); 
				\draw[red, thick, -latex] (8.45,1.85) -- (8.45,5.15); 
				\draw[black, very thick, loosely dotted, -latex] (9.3,1.85) -- (9.3,3.15); 

				\draw[rounded corners] (8.6,9.7) rectangle (10,10.3);
				\draw (9.3,10) node [] {PTIME};
				\draw[rounded corners] (8,7.2) rectangle (9.8,7.8);
				\draw (8.9,7.5) node [] {PTIME[E]};
				\draw[red, thick, -latex] (9.3,7.85) -- (9.3,9.65); 
				\draw[rounded corners] (8.6,3.2) rectangle (10.4,3.8);
				\draw (9.5,3.5) node [] {PTIME[H]};
				\draw[red, thick, -latex] (9.3,3.8) -- (9.3,7.1); 
				
			\end{tikzpicture}
		\end{adjustbox}
	\end{center}
	\caption{\label{fig-summary}The relationship of FO, LFP, PTIME, Datalog and its variants. Datalog$^r$ denotes revised Datalog. pos-$\mathcal{L}$ denotes the positive fragment of $\mathcal{L}$. $\exists\mathcal{L}$ denotes the existential fragment of $\mathcal{L}$. $\mathcal{L}$[FO] denotes the set of $\mathcal{L}$ formulas that are first-order definable. $\mathcal{L}$[H] (resp., $\mathcal{L}$[E]) denotes the set of $\mathcal{L}$ formulas (or problems computable in $\mathcal{L}$) that are preserved under homomorphisms (resp., extensions).
	The blue arrow shows the containment relationship on Datalog and its variants. The red arrow shows the relationship about the expressive power. The solid arrow implies that the relationship is strict, and the dashed bidirectional arrow implies the equality relationship. The black dotted arrow means whether the relationship is strict is still open.}
\end{figure}

Revised Datalog (Datalog$^r$) is an extension of Datalog, where universal quantification over intensional relations is allowed in the body of rules.  Abiteboul and Vianu first introduced the idea that the body of a rule in Datalog can be universally quantified~\cite{abiteboul1991datalog}. The author of the paper showed that Datalog$^r$ equals LFP on all finite structures~\cite{feng2012}. In the paper, we study the expressive power of Datalog$^r$ on problems with closure properties, i.e., closed under substructures (or extensions). We conclude that a Datalog$^r$ formula is equivalent to a first-order formula iff it is equivalent to a bounded Datalog$^r$ formula. As the main result of the paper, we show that Datalog$^r$ cannot define all the problems that are in PTIME and closed under substructures. Since Datalog$^r$ equals LFP, and the complement of a substructure-closed problem is extension-closed, as a corollary, LFP cannot define all the extension-closed problems that are in PTIME. This result contributes the strict containment $\mathrm{LFP[E]} \subsetneq \mathrm{PTIME[E]}$ in Fig.~\ref{fig-summary}.  A technique of tree encodings for arbitrary structures is used in the proof. For an arbitrary set of structures $\mathcal{K}\in$ EXPTIME, we can encode them into a set of substructure-closed structures $\mathcal{K}'$, where the tree used to encode the structure in $\mathcal{K}$ is exponentially larger. Therefore, $\mathcal{K}'$ is in PTIME. For every structure in $\mathcal{K}'$, there is a characteristic structure $S_{\mathbf{T}}$ of it such that they are equivalent with respect to Datalog$^r$-transformations. Since $S_{\mathbf{T}}$ can be computed from the structure in $\mathcal{K}$ in logspace, this implies that $\mathcal{K}$ is also in PTIME, contrary to the time hierarchy theorem. Fig.~\ref{fig-tree-encoding} shows the sketch of the proof.

\begin{figure}
	\begin{center}
	\begin{adjustbox}{scale = 0.9}
		\begin{tikzpicture}[font=\footnotesize]
			\draw[rounded corners] (-0.8,0.2) rectangle (1.8,1.2);
			\draw (0.5,0.9) node [] {Arbitrary};
			\draw (0.5,0.5) node [] {$\mathcal{K}\in$ EXPTIME};
			
			\draw (3.15,0.9) node [] {Tree encodings};
			\draw[thick](2,0.7) -- (4.4,0.7);
			\draw[thick](4.3,0.8) -- (4.4,0.7) -- (4.3,0.6);
			
			\draw[rounded corners] (4.5,0.2) rectangle (7.6,1.2);
			\draw (6.05,0.9) node [] {Substructure closed};
			\draw (6.05,0.5) node [] {$\mathcal{K}'\in$ PTIME};
			
			\draw (8.7,0.95) node [] {Datalog$^r$};
			\draw (8.7,0.45) node [] {equivalent};
			\draw[thick](7.7,0.7) -- (9.7,0.7);
			\draw[thick](9.6,0.8) -- (9.7,0.7) -- (9.6,0.6);

			\draw[rounded corners] (9.8,0.2) rectangle (13.2,1.2);
			\draw (11.5,0.9) node [] {Class of characteristic};
			\draw (11.5,0.5) node [] {structures $S_{\mathbf{T}}$};
			
			\draw (6.3,-0.4) node [] {Logspace computable};
			\draw [thick] (0.5,0.1) -- (0.5,-0.6) -- (11.5,-0.6) -- (11.5,0.1);
			\draw [thick] (11.35,-0.1) -- (11.5,0.1) -- (11.65,-0.1);
		\end{tikzpicture}
	\end{adjustbox}
	\end{center}
	\caption{The idea of the proof for the nondefinability of Datalog$^r$.\label{fig-tree-encoding}}
\end{figure}

The paper is organized as follows: In Section~\ref{sec-pre}, we give the basic definitions and notations. 
In Section~\ref{sec-redatalog-closed}, we recall invariant relations on perfect binary trees, and introduce the technique of tree encodings for arbitrary structures.  And we prove the nondefinability results for Datalog$^r$ on substructure-closed problems. Finally, we conclude the paper in Section~\ref{sec-concl}.

\section{Preliminaries}\label{sec-pre}
Let $\tau=\{\mathbf{c}_{1},\dots,\mathbf{c}_{m},P_{1},\dots,P_{n}\}$ be a vocabulary, where $\mathbf{c}_{1},\dots,\mathbf{c}_{m}$ are constant symbols and $P_{1},\dots,P_{n}$ are relation symbols. A $\tau$-structure is a tuple 
$\mathbf{A}=\langle A,\mathbf{c}_{1}^{A},\dots,\mathbf{c}_{m}^{A},P_{1}^{A},\dots,P_{n}^{A}\rangle$
where $A$ is the domain, and $\mathbf{c}_{1}^{A},\dots,\mathbf{c}_{m}^{A}$, $P_{1}^{A},\dots,P_{n}^{A}$ are interpretations of the constant and relation symbols over $A$, respectively. We assume the equality relation ``='' is contained in every vocabulary, and omit the superscript ``$A$'' when it is clear from context.
We call $\mathbf{A}$ \textit{finite} if its domain $A$ is a finite set. Unless otherwise stated, all structures considered in this paper are finite. We use $arity(R)$ to denote the arity of a relation $R$, and use ``$|\ |$'' to indicate the cardinality of a set or the arity of a tuple, e.g., $|A|$ denotes the cardinality of $A$ and $|(x_1,x_2,x_3)|=3$.
A finite structure is \textit{ordered} if it is equipped with a linear order relation ``$\leq$'', and the successor relation ``$\mathrm{SUCC}$'', the constants ``\textbf{min}'' and ``\textbf{max}''for the minimal and maximal elements, respectively, with respect to ``$\leq$''.
Let $\mathbf{A}=\langle A,\mathbf{c}_{1}^{A},\dots,\mathbf{c}_{m}^{A},P_{1}^{A},\dots,P_{n}^{A}\rangle$ and $\mathbf{B}=\langle B,\mathbf{c}_{1}^{B},\dots,\mathbf{c}_{m}^{B},P_{1}^{B},\dots,P_{n}^{B}\rangle$ be two structures. If $B\subseteq A$, $\mathbf{c}_{i}^{A} = \mathbf{c}_{i}^{B}$ $(1\leq i \leq m)$, and $P_{j}^{B} = P_{j}^{A} \cap B^{arity(P_j)}$ $(1\leq j \leq n)$, then we say that $\mathbf{B}$ is a \textit{substructure} of $\mathbf{A}$, and $\mathbf{A}$ is an \textit{extension} of $\mathbf{B}$.

An $r$-ary \textit{global relation} $R$ of a vocabulary $\tau$ is a mapping that assigns to every $\tau$-structure $\mathbf{A}$ an $r$-ary relation $R^{A}$ over $A$ such that for every isomorphism $\pi: \mathbf{A} \simeq \mathbf{B}$ and every $a_1,\dots,a_r\in A$, $\mathbf{A} \models R^{A}a_1\dots a_r$ iff $\mathbf{B} \models R^{B}\pi(a_1)\dots \pi(a_r)$.
A \textit{query} is a global relation. We say that a query $\mathcal{Q}$ is expressible in a logic $\mathcal{L}$ if there is an $\mathcal{L}$-formula that defines $\mathcal{Q}$. 
Two formulas are \textit{equivalent} if they define the same query.
Given two logics $\mathcal{L}_1$ and $\mathcal{L}_2$, we use $\mathcal{L}_1\leq \mathcal{L}_2$ to denote that every $\mathcal{L}_1$-formula is equivalent to an $\mathcal{L}_2$-formula. If $\mathcal{L}_1\leq \mathcal{L}_2$ and $\mathcal{L}_2\leq \mathcal{L}_1$, then we denote it by $\mathcal{L}_1\equiv \mathcal{L}_2$.

Suppose that a relation symbol $X$ occurs positively in $\varphi(\bar{x})$ and $|\bar{x}|=arity(X)$. Given a structure $\mathbf{A}$, we can define a monotonic sequence $X_0, X_1, X_2, \dots$ where $X_0 = \emptyset$, and $X_{i+1} = \{\bar{a} \mid (\mathbf{A}, X_i) \vDash \varphi[\bar{a}] \}$ for $i\geq 0$. Since $\mathbf{A}$ is finite, the sequence $X_0, X_1, X_2, \dots$ will eventually reach a fixpoint.

\begin{definition}
	The least fixpoint logic $\mathrm{LFP}$ is an extension of first-order logic by adding the following rule~\cite{ebbinghaus1995}:
	\begin{itemize}
		\item If $\varphi$ is an $\mathrm{LFP}$ formula, $X$ occurs positively in $\varphi$, and $|\bar{x}|=|\bar{u}|=arity(X)$, then 
		$[\mathrm{LFP}_{\bar{x},X}\varphi]\bar{u}$ is an $\mathrm{LFP}$ formula.
	\end{itemize}
\end{definition}

Given an LFP formula $[\mathrm{LFP}_{\bar{x},X}\varphi]\bar{u}$, for any structure $\mathbf{A}$ and $\bar{a} \in A$, we have $\mathbf{A} \vDash [\mathrm{LFP}_{\bar{x},X}\varphi]\bar{a}$ iff $\bar{a}$ is in the fixpoint of the sequence induced by $X$ and $\varphi$ on $\mathbf{A}$. 

\begin{proposition}~\cite{immerman1982relational,vardi1982complexity}
$\mathrm{LFP}$ captures $\mathrm{PTIME}$ on ordered finite structures.
\end{proposition}

\begin{definition}\label{def:defnOFdatalog}
	Let $\tau$ be a vocabulary. A $\mathrm{Datalog}$ program $\mathrm{\Pi}$ over $\tau$ is a finite set of rules of the form
	\[
	\beta\leftarrow\alpha_{1},\dots,\alpha_{l}
	\]
	where $l\geq0$ and
	\begin{description}
		\item {(1)} each $\alpha_{i}$ is either an atomic formula or a negated
		atomic formula,
		\item {(2)} $\beta$ is an atomic formula $R\bar{x}$, where $R$ doesn't
		occur negatively in any rule of $\mathrm{\Pi}$.
	\end{description}
\end{definition}

$\beta$ is the head of the rule and the sequence $\alpha_{1},\dots,\alpha_{l}$ constitute the body. Every relation symbol occurring in the head of some rule of $\mathrm{\Pi}$ is intensional, and the other symbols in $\tau$ are extensional. We use $(\tau,\mathrm{\Pi})_{\mathrm{int}}$ and $(\tau,\mathrm{\Pi})_{\mathrm{ext}}$ to denote the set of intensional and extensional symbols, respectively.
We also allow 0-ary relation symbols. If $Q$ is a 0-ary relation, its value is from $\{\emptyset,\{\emptyset\}\}$. $Q=\emptyset$ means that $Q$ is FALSE and $Q=\{\emptyset\}$ means that $Q$ is TRUE. 
We use the least fixpoint semantics for Datalog programs. 
A Datalog formula has the form $(\mathrm{\Pi},P)\bar{x}$, where $P$ is an $r$-ary intensional relation symbol and $\bar{x}=x_{1},\dots,x_{r}$ are variables that do not occur in $\mathrm{\Pi}$. For a $(\tau,\mathrm{\Pi})_{\mathrm{ext}}$-structure $\mathbf{A}$ and $\bar{a}=a_{1},\dots,a_{r}\in A$,
\[
\mathbf{A}\models(\mathrm{\Pi},P)\bar{x}[\bar{a}]\text{ iff }(a_{1},\dots,a_{r})\in P_{(\infty)},
\]
where $P_{(\infty)}$ is the least fixpoint for relation $P$ when $\mathrm{\Pi}$ is evaluated on $\mathbf{A}$. If $P$ is 0-ary, then $\mathbf{A}\models(\mathrm{\Pi},P)\text{ iff }P_{(\infty)}=\{\emptyset\}$.

\section{$\mathrm{Datalog}^r$ on problems with closure properties}\label{sec-redatalog-closed}
\subsection{Revised Datalog programs}

\begin{definition}
	In Definition~\ref{def:defnOFdatalog}, if we replace Condition (1) by
	\begin{description}
		\item {(1$'$)} each $\alpha_{i}$ is either an atomic formula, or a negated
		atomic formula, or a formula $\forall\bar{y}R\bar{y}\bar{z}$, where $R$ occurs in the head of some rule,
	\end{description}
	then we call this logic program revised $\mathrm{Datalog}$ program, denoted by $\mathrm{Datalog}^r$.
\end{definition}

\begin{example} 
	Let $G =\langle V,E \rangle$ be a directed acyclic graph, and the set of nodes $V$ partitioned into two disjointed sets $V_{\mathrm{uni}}$ and $V_{\mathrm{exi}}$. The nodes in $V_{\mathrm{uni}}$ (resp., $V_{\mathrm{exi}}$) are universal (resp., existential). The notion of alternating path is defined recursively.
	There is an alternating path from $\mathbf{s}$ to $\mathbf{t}$ in $G$ if
	\begin{itemize}
		\item $\mathbf{s} = \mathbf{t}$; or 
		\item $\mathbf{s}\in V_{\mathrm{exi}}$, $\exists x\in V$ such that $(\mathbf{s},x)\in E$ and
		there is an alternating path from $x$ to $\mathbf{t}$; or
		\item $\mathbf{s}\in V_{\mathrm{uni}}$, $\exists x\in V$ such that $(\mathbf{s},x)\in E$, and $\forall y\in V$, if $(\mathbf{s},y)\in E$ then there
		is an alternating path from $y$ to $\mathbf{t}$.
	\end{itemize}
	The alternating graph accessibility problem is defined as follows:
	\begin{description}
		\item [{Input:}] A directed acyclic graph $G=\langle V_{\mathrm{uni}}\cup V_{\mathrm{exi}}, E \rangle$ and two nodes $\mathbf{s}, \mathbf{t}$.
		\item [{Output:}] Yes if there is an alternating path from $\mathbf{s}$ to $\mathbf{t}$ in $G$, otherwise no.
	\end{description}
	This problem is $\mathrm{P}$-complete~\cite{immer1981number}.  
	The following $\mathrm{Datalog}^r$ program $\mathrm{\Pi}$ defines the alternating graph accessibility problem
	\vspace{-1.5em}
	\begin{center}
		\begin{minipage}{0.4\textwidth}
			\begin{align*}
				P_{\mathrm{alt}}xy & \leftarrow  x=y;\\
				P_{\mathrm{alt}}xy & \leftarrow  \neg V_{\mathrm{uni}}x,Exz,P_{\mathrm{alt}}zy;\\
				P_{\mathrm{uni}}x  & \leftarrow  V_{\mathrm{uni}}x,Exy;\\
				Qxzy 	  & \leftarrow  P_{\mathrm{uni}}x, \neg Exz;
			\end{align*} 
		\end{minipage}
		\begin{minipage}{0.4\textwidth}
			\begin{align*}
				Qxzy	  & \leftarrow  P_{\mathrm{uni}}x,Exz,P_{\mathrm{alt}}zy;\\
				P_{\mathrm{alt}}xy & \leftarrow  P_{\mathrm{uni}}x,\forall zQxzy;\\
				P 		  & \leftarrow  P_{\mathrm{alt}}\mathbf{s}\mathbf{t}.
			\end{align*}
		\end{minipage}
	\end{center}
	We have $(\tau,\mathrm{\Pi})_{\mathrm{int}}=\{P_{\mathrm{alt}},Q,P_{\mathrm{uni}},P\}$ and $(\tau,\mathrm{\Pi})_{\mathrm{ext}}=\{E,V_{\mathrm{uni}},\mathbf{s},\mathbf{t}\}$.
	The relation $P_{\mathrm{uni}}$ saves the node in $V_{\mathrm{uni}}$ that has a successor. The relation $P_{\mathrm{alt}}$ saves the pair $(x,y)$ such that there is an alternating path from $x$ to $y$.
	We use $Qxzy$ to denote that for any $x \in P_{\mathrm{uni}}$, either there is no edge from $x$ to $z$, or there is an alternating path from $z$ to $y$. 
	Consider the $\mathrm{Datalog}^r$ formula $(\mathrm{\Pi},P)$, for any directed acyclic $(\tau,\mathrm{\Pi})_{\mathrm{ext}}$-structure $\mathbf{A}$, it is easy to check that $\mathbf{A}\models(\mathrm{\Pi},P)$
	iff there is an alternating path from $\mathbf{s}$ to $\mathbf{t}$.
\end{example}

The Datalog formulas are preserved under extensions~\cite{Dawar2021extension}, i.e., if a structure $\mathbf{B}$ satisfies a Datalog formula $\varphi$ and $\mathbf{A}$ is an extension of $\mathbf{B}$, then $\mathbf{A}$ also satisfies $\varphi$. A directed acyclic graph with an alternating path from $\mathbf{s}$ to $\mathbf{t}$ can be extended to a directed acyclic graph without any alternating path from $\mathbf{s}$ to $\mathbf{t}$ by adding new nodes. So Datalog cannot define the alternating graph accessibility problem, which implies that $\mathrm{Datalog}^r$ is strictly more expressive than $\mathrm{Datalog}$. 
Allowing universal quantification over intensional relations is essential for $\mathrm{Datalog}^r$ to increase its expressive power. With the help of it, every FO(LFP) formula can be transformed into an equivalent $\mathrm{Datalog}^r$ formula.

\begin{proposition}\cite{feng2012}\label{prop-datalogrequivlfp}
	$\mathrm{Datalog}^r\equiv \mathrm{LFP}$ on all finite structures.
\end{proposition}

A Datalog program is \textit{positive} if no negated atomic formula occurs in the body of any rule. A Datalog formula  $(\mathrm{\Pi},P)\bar{t}$ is \textit{bounded} if there is an $n\geq 0$ such that $P_{(n)}=P_{(\infty)}$ for all structures. 
A bounded (positive) Datalog formula is equivalent to an existential (positive) first-order formula, and vice versa~\cite{ebbinghaus1995}.
Furthermore, a positive Datalog formula is bounded iff it is equivalent to a first-order formula. The statement is false for all Datalog formulas. There is an unbounded Datalog formula that is equivalent to an FO formula, but no bounded Datalog formula is equivalent to it~\cite{ajtai1994datalog}. Unlike Datalog, if an unbounded $\mathrm{Datalog}^r$ formula is equivalent to an FO formula, then it must be equivalent to a bounded $\mathrm{Datalog}^r$ formula.

\begin{proposition}
	A $\mathrm{Datalog}^r$ formula is equivalent to a first-order formula iff it is equivalent to a bounded $\mathrm{Datalog}^r$ formula.
\end{proposition}
\begin{proof}
	Suppose that a $\mathrm{Datalog}^r$ formula is equivalent to a first-order formula $\varphi$. Using the method in~\cite{feng2012} we can construct a bounded $\mathrm{Datalog}^r$ formula that is equivalent to $\varphi$. For the other direction, the proof in~\cite{ebbinghaus1995} which shows that every bounded $\mathrm{Datalog}$ formula is equivalent to an FO formula remains valid for bounded $\mathrm{Datalog}^r$ formulas.
\end{proof}

\subsection{Invariant relations on perfect binary trees}\label{sec-invar}
In~\cite{lindell1991analysis}, Lindell introduced invariant relations that are defined on perfect binary trees, and showed that there are queries computable in PTIME but not definable in LFP.

A perfect binary tree is a binary tree in which all internal nodes have two children and all leaf nodes are in the same level. Let $T=\langle V,E,\mathbf{root}\rangle $ be a perfect binary tree, where $V$ is the set of nodes, $E$ is the set of edges and $\mathbf{root}$ is the root node. 
Suppose that $R$ is an $r$-ary relation on $V$ and $f$ is an automorphism of $T$. Given a tuple $\bar{a}=(a_{1},\dots,a_{r})\in R$, we write $f(\bar{a})=(f(a_{1}),\dots,f(a_{r}))$ and $f[R]=\{(f(a_{1}),\dots,f(a_{r}))\mid(a_{1},\dots,a_{r})\in R\}$. We say that $R$ is an invariant relation if for every automorphism $f$, $R=f[R]$. It is easily seen that the equality $=$ and $E$ are invariant relations.


\begin{lemma}
	\label{lem:lem1} If $R_{1}$ and $R_{2}$ are $r$-ary invariant relations, then $\neg R_{1}$, $R_{1}\cap R_{2}$ and $R_{1}\cup R_{2}$ are also invariant relations.
\end{lemma}
\begin{proof}
	It is easy to check that $\neg R_{1}$ is an invariant relation. Let $f$ be an arbitrary automorphism and $f^{-1}$ its inverse. To show that $R_{1}\cap R_{2}$ is an invariant relation, we have for any $\bar{a}$
	\[
		\bar{a}\in R_{1}\cap R_{2} \Longleftrightarrow  \bar{a}\in f[R_{1}]\cap f[R_{2}] \Longleftrightarrow f^{-1}(\bar{a})\in  R_{1}\cap R_{2}\Longleftrightarrow \bar{a}\in f[R_{1}\cap R_{2}].
	\]
	To show that $R_{1}\cup R_{2}$ is an invariant relation, we have for any $\bar{a}$
	\[
	\bar{a}\in R_{1}\cup R_{2} \Longleftrightarrow  \bar{a}\in f[R_{1}]\cup f[R_{2}] \Longleftrightarrow f^{-1}(\bar{a})\in  R_{1}\cup R_{2}\Longleftrightarrow \bar{a}\in f[R_{1}\cup R_{2}].
	\]
	This completes the proof. \qed
\end{proof}

\begin{lemma}\label{lem:lem2}
	Suppose that $R$ is an $r$-ary invariant relation,
	$R'$ is a $k$-ary invariant relation
	and $g$ is a permutation of $\{1,\dots,r\}$. Define
	\[
	\begin{aligned}
		R_{1} & = \{(a_{g(1)},\dots,a_{g(r)})\mid(a_{1},\dots,a_{r})\in R\},\\
		R_{2} & = \{(a_{1},\dots,a_{r},b_{1},\dots,b_{k})\mid (a_{1},\dots,a_{r})\in R \text{ and }(b_{1},\dots,b_{k})\in R' \}.
	\end{aligned}
	\]
	Then $R_{1}$ and $R_{2}$ are also invariant relations.
\end{lemma}
\begin{proof}
	Let $f$ be an arbitrary automorphism, $f^{-1}$ its inverse, and $g^{-1}$ the inverse of $g$. 
	For a relation $P$, let $g[P]$ be the relation obtained by permuting the tuples in $P$ with respect to $g$. It is easily seen that $f(g[P]) = g(f[P])$. To show that $R_{1}=f[R_{1}]$, we have for any $\bar{a}$
	\[
		\bar{a}\in R_{1} \Longleftrightarrow  g^{-1}(\bar{a})\in R \Longleftrightarrow g^{-1}(\bar{a})\in f[R]  \Longleftrightarrow \bar{a}\in g(f[R]) \Longleftrightarrow \bar{a}\in f(g[R]) \Longleftrightarrow \bar{a}\in f[R_1] .
	\]
	
	To show that $R_{2}=f[R_{2}]$, we have for any $\bar{a}$
	\[
	\begin{aligned}
		\bar{a}\bar{b}\in R_{2} & \Longleftrightarrow  \bar{a}\in R\text{ and }\bar{b}\in R'  \Longleftrightarrow \bar{a}\in f[R]\text{ and }\bar{b}\in f[R']\\
		 & \Longleftrightarrow f^{-1}(\bar{a}) \in R\text{ and } f^{-1}(\bar{b})\in R' \Longleftrightarrow f^{-1}(\bar{a} \bar{b})\in  R_{2} \Longleftrightarrow \bar{a}\bar{b}\in f[R_{2}].
	\end{aligned}
	\]
	 This completes the proof.\qed
\end{proof}

\begin{lemma}
	\label{lem:lem3}
	Suppose that $R$ is a $(k+r)$-ary invariant relation.
	Define
	\[
	\begin{array}{cl}
		R_{1}= & \{(a_{1},\dots,a_{r})\mid(b_{1},\dots,b_{k},a_{1},\dots,a_{r})\in R\text{ for all nodes }b_{1},\dots,b_{k}\}\\
		R_{2}= & \{(a_{1},\dots,a_{r})\mid\exists b_{1},\dots,b_{k} \text{ such that }(b_{1},\dots,b_{k},a_{1},\dots,a_{r})\in R\}.
	\end{array}
	\]
	Then $R_{1}$ and $R_{2}$ are also invariant relations.
\end{lemma}
\begin{proof}
	Let $f$ be an arbitrary automorphism, and $f^{-1}$ its inverse.
	To show that $R_{1}=f[R_{1}]$, we have for any $\bar{a}$
	\[
	\begin{aligned}
	\bar{a}\in R_{1} & \Longleftrightarrow \bar{b}\bar{a}\in R\text{ for any tuple }\bar{b} \Longleftrightarrow \bar{b}\bar{a}\in f[R]\text{ for any tuple }\bar{b} \\ 
	& \Longleftrightarrow  \bar{b}f^{-1}(\bar{a})\in R \text{ for any tuple }\bar{b} \Longleftrightarrow f^{-1}(\bar{a})\in R_{1} \Longleftrightarrow \bar{a}\in f[R_{1}]. 
	\end{aligned}
	\]

	To show that $R_{2}=f[R_{2}]$, we have for any $\bar{a}$
	\[
	\begin{aligned}
		\bar{a}\in R_{2} & \Longleftrightarrow \exists \bar{b}\text{ such that }\bar{b}\bar{a}\in R \Longleftrightarrow \exists \bar{b}\text{ such that }\bar{b}\bar{a}\in f[R] \\ 
		& \Longleftrightarrow \exists \bar{b'}\text{ such that }\bar{b'}f^{-1}(\bar{a})\in R \Longleftrightarrow  f^{-1}(\bar{a})\in R_{2} \Longleftrightarrow \bar{a}\in f[R_{2}].
	\end{aligned}
	\]
	This completes the proof.\qed
\end{proof}

Let $a,b$ be two nodes of a perfect binary tree $T$, we use $a\barwedge b$ and $d(a)$
to denote the least common ancestor of $a,b$ and the depth of $a$,
respectively. For example, in Fig.~\ref{fig-binarytree} there is a perfect binary tree in which $d(\mathbf{root})=0$, $d(a)=1$, $d(c)=d(e)=2$, and $c\barwedge e=\mathbf{root}$.
\begin{figure}
	\begin{center}
		\begin{tikzpicture}[level distance=0.8cm,
			level 1/.style={sibling distance=3cm},
			level 2/.style={sibling distance=1.5cm}]
			\node {\textbf{root}}
			child {node {a}
				child {node {c}}
				child {node {d}}
			}
			child {node {b}
				child {node {e}}
				child {node {f}}
			};
			\draw[dashed] (-0.4,-0.15) -- (3.3,-0.15);
			\draw[dashed] (-1.6,-0.95) -- (3.3,-0.95);
			\draw[dashed] (-2.3,-1.75) -- (3.3,-1.75);
			\node at (3.3,0) {0};
			\node at (3.3,-0.8) {1};
			\node at (3.3,-1.6) {2};
		\end{tikzpicture}
	\end{center}
	\caption{\label{fig-binarytree}A perfect binary tree of depth 3.}
\end{figure}

Let $(a_{1},\dots,a_{r})$ be an $r$-ary tuple of nodes, its characteristic tuple is defined as
\[
\begin{array}{rr}
	(a_{1},\dots,a_{r})^{*}= & (d(a_{1}),d(a_{1}\barwedge a_{2}),\dots,d(a_{1}\barwedge a_{r}),\,\\
	& d(a_{2}),d(a_{2}\barwedge a_{3}),\dots,d(a_{2}\barwedge a_{r}),\,\\
	& \dots,d(a_{r}))
\end{array}
\]
which is a $\frac{r(r+1)}{2}$-ary tuple of numbers.
Let $R$ be an invariant relation, the characteristic relation of $R$ is defined to be
\[
R^{*}=\{(a_{1},\dots,a_{r})^{*}\mid(a_{1},\dots,a_{r})\in R\}.
\]

\begin{proposition}\cite{lindell1991analysis} 
	Let $\bar{a}=(a_{1},\dots,a_{r})$ and $\bar{b}=(b_{1},\dots,b_{r})$ be two tuples, and $R$ an $r$-ary invariant relation of a perfect binary tree $T$.
	\begin{itemize}
		\item  $(a_{1},\dots,a_{r})^{*}=(b_{1},\dots,b_{r})^{*}$ iff there is an automorphism $f$ of $T$ such that $f(\bar{a})=\bar{b}$.
		\item If $(a_{1},\dots,a_{r})^{*}=(b_{1},\dots,b_{r})^{*}$, then $\bar{a}\in R$ iff $\bar{b}\in R$.
	\end{itemize}
	For any two invariant relations $R_{1}$ and $R_{2}$, $R_{1}=R_{2}$
	iff $R_{1}^{*}=R_{2}^{*}$.
\end{proposition}

\subsection{Tree encodings and characteristic structures}\label{sec-treeencode}

This section is devoted to the definitions of tree encodings and characteristic structures, and the propositions about the equivalent relationship between them on Datalog$^r$ programs, which will be used in the next section.

\begin{definition} \label{fz1}
	Let $T$ be a perfect binary tree, $R$ an $r$-ary relation on $T$. $R$ is a saturated relation if for any nodes $a_{1},\dots,a_{r}$, $b_{1},\dots,b_{r}$, whenever $d(a_{i})=d(b_{i})\,(1\leq i\leq r)$, then
	$(a_{1},\dots,a_{r})\in R \mbox{ iff } (b_{1},\dots,b_{r})\in R$.
\end{definition}

The following proposition can be proved easily from the definitions of invariant relations and saturated relations.
\begin{proposition}
	A saturated relation is also an invariant relation.
\end{proposition}

From now on we make the assumption: $\tau$ is the vocabulary $\{R_{1},\dots,R_{k}\}$, and $\tau'=\tau\cup\{\mathbf{root},E\}$, where $\mathbf{root}$ is a constant symbol and $E$ is a binary relation symbol that is not in $\tau$. We define a class of $\tau'$-structures
\[
\begin{array}{r}
	\mathcal{T}=\{\langle V,\mathbf{root},E,R_{1},\dots,R_{k}\rangle \mid\langle V,E,\mathbf{root}\rangle \textrm{ is a perfect binary tree}, \\
	R_{1},\dots,R_{k} \textrm{ are saturated relations on it} \}.
\end{array}
\]

\begin{definition}
	Let $\mathbf{A}=\langle \{0,1,\dots,h-1\},R_{1}^{A},\dots,R_{k}^{A}\rangle$
	be a $\tau$-structure. The tree encoding of $\mathbf{A}$ is a $\tau'$-structure $C(\mathbf{A})=\langle V,\mathbf{root},E,R_{1}^{T},\dots,R_{k}^{T}\rangle \in \mathcal{T}$, such that $\langle V,E,\mathbf{root}\rangle$ is a perfect binary tree of depth $h$, and for any relation symbol $R_{i}\,(1\leq i\leq k)$ and any nodes $a_{1},\dots,a_{r_i}\in V$,
	\[ 
	C(\mathbf{A})\models R_{i}^{T}a_{1}\cdots a_{r_{i}}
	\mbox{ iff } \mathbf A\models R_{i}^{A}d(a_{1})\cdots d(a_{r_{i}})
	\]
	where $r_i$ is the arity of $R_i$, and $d(a_{j})\,(1\leq j \leq r_i)$ is the depth of $a_{j}$.
\end{definition}

Roughly speaking, $C(\mathbf{A})$ encodes $\mathbf{A}$ in a tree, but its size is exponentially larger. Conversely, given a $\tau'$-structure $\mathbf{T}=\langle V,\mathbf{root},E,R_{1}^{T},\dots,R_{k}^{T}\rangle\in \mathcal{T}$, we can compute the $\tau$-structure $\mathbf{A}$ encoded by $\mathbf{T}$ as follows: 
\begin{description}
	\item[(1)] The domain is $\{0,\dots, h-1\}$, where $h$ is the depth of $\mathbf{T}$;
	\item[(2)] For each $i=1,\dots,k$,
	\[ 
	R^A_i = \{(d(a_1), \dots, d(a_{r_i}))\mid \exists a_1,\dots, a_{r_i}\in V  \text{ such that } \mathbf{T}\models R_i^Ta_1\cdots a_{r_i}\}.
	\]
\end{description}

We use $C^{-1}(\mathbf{T})$ to denote the corresponding $\tau$-structure $\mathbf{A}$ encoded by $\mathbf{T}$.
Let $\mathrm{FUL}_m = V^m$ be a relation of arity $m$, where $m\geq 1$ and $V$ is the domain of $\mathbf{T}$. Define the vocabulary
\[
\sigma =\{\mathbf{0},\mathrm{SUCC},R_{\neq},R_{\neg e},\mathrm{FUL}_{m}^{*}\} \cup \{R_{1}^{*},\dots,R_{k}^{*},(\neg R_{1})^{*},\dots,(\neg R_{k})^{*}\}
\]
where $\mathrm{FUL}_{m}^{*}$ has arity $\frac{m(m+1)}{2}$, $R_{\neq}$ and $R_{\neg e}$ have arity 3, $R_i^*$ or $(\neg R_{1})^{*}$ has arity $\frac{r_i(r_i+1)}{2}$ ($1\leq i \leq k$ and $r_i$ is the arity of $R_i$).

\begin{definition}\label{fz4}
	Given a $\tau'$-structure $\mathbf{T}=\langle V,\mathbf{root},E,R_{1},\dots,R_{k}\rangle\in\mathcal{T}$, the characteristic structure $S_{\mathbf{T}}$ of $\mathbf{T}$ is a $\sigma$-structure
	\[
	\langle\{0,1,\dots,h-1\},\mathbf{0},\mathrm{SUCC},R_{\neq},R_{\neg e},\mathrm{FUL}_{m}^{*}, R_{1}^{*},\dots,R_{k}^{*},(\neg R_{1})^{*},\dots,(\neg R_{k})^{*}\rangle
	\]
	where $h$ is the depth of $\mathbf{T}$, $\mathbf{0}$ is a constant interpreted by $0$, $\mathrm{SUCC}$ is the successor relation on the domain, and $R_{\neq},R_{\neg e},\mathrm{FUL}_{m}^{*}$, $R_{1}^{*},\dots,R_{k}^{*}$, $(\neg R_{1})^{*},\dots,(\neg R_{k})^{*}$ are the characteristic relations of $\neq,(\neg E)$, $\mathrm{FUL}_{m},R_{1},\dots,R_{k}$, $\neg R_{1},\dots,\neg R_{k}$, respectively.
\end{definition}

In the following we show that for every Datalog$^r$ program $\mathrm{\Pi}$ on the tree encodings, there is a Datalog$^r$ program $\mathrm{\Pi}^{*}$ on the corresponding characteristic structures such that $\mathrm{\Pi}^{*}$ simulates $\mathrm{\Pi}$. 
More precisely, $\mathrm{\Pi}^{*}$ operates the characteristic relations of the relations in $\mathrm{\Pi}$.
Let $\mathrm{\Pi}=\{\gamma_{1},\dots,\gamma_{s}\}$ be a Datalog$^r$ program
on $\mathcal{T}$.  Suppose $X_{1},\dots,X_{w}$ are all intensional
relation symbols in $\mathrm{\Pi}$ and for each rule $\gamma_{i}$, let $n_{\gamma_{i}}$
be the number of free variables occurring in $\gamma_{i}$. Set
\[
m= \max\{n_{\gamma_{1}},\dots,n_{\gamma_{s}},\mathrm{arity}(R_{1}),\dots,\mathrm{arity}(R_{k}),\mathrm{arity}(X_{1}),\dots,\mathrm{arity}(X_{w})\}.
\]
We shall construct, based on $\mathrm{\Pi}$, a Datalog$^r$ program $\mathrm{\Pi}^{*}$ such that for any Datalog$^r$ formula $(\mathrm{\Pi},P)$, there exists a Datalog$^r$ formula $(\mathrm{\Pi}^{*},P^{*})$, and $\mathbf{T}\models(\mathrm{\Pi},P)$ iff $S_{\mathbf{T}}\models(\mathrm{\Pi}^{*},P^{*})$
for any $\mathbf{T}\in\mathcal{T}$, where $P$ and $P^{*}$ are 0-ary.

Every element of $\mathbf{T}$ is a node of a perfect binary tree,
while every element of $S_{\mathbf{T}}$ is a number which can be treated as the depth of some node. Hence, for each variable $x$ in $\mathrm{\Pi}$, we introduce a new variable $i_{x}$, and for any two variables $x_1, x_2$ in $\mathrm{\Pi}$, we introduce a new variable $i_{x_1\barwedge x_2}$.
For a tuple of variables $\bar{x}=x_{1}\cdots x_{r}$, we use the following abbreviations:
\[
\begin{aligned}
	(\bar{x})^{*} & = i_{x_{1}}i_{x_{1}\barwedge x_{2}}\cdots i_{x_{1}\barwedge x_{r}}i_{x_{2}}i_{x_{2}\barwedge x_{3}}\cdots i_{x_{r-1}\barwedge x_{r}}i_{x_{r}},\\
	\forall(\bar{x})^{*} &  =  \forall i_{x_{1}}\forall i_{x_{1}\barwedge x_{2}}\cdots\forall i_{x_{1}\barwedge x_{r}}\forall i_{x_{2}}\forall i_{x_{2}\barwedge x_{3}}\cdots\forall i_{x_{r-1}\barwedge x_{r}}\forall i_{x_{r}}.
\end{aligned}
\]
Without loss of generality, we treat $i_{u \barwedge v}$ and $i_{v \barwedge u}$ as the same variable. 
Additionally, we assume that the 0-ary relation is also an invariant relation, and the characteristic relation of a 0-ary relation is itself.

First we construction a quasi-Datalog$^r$ program $\mathrm{\Pi}'$ as follows. For each rule $\beta\leftarrow\alpha_{1},\dots,\alpha_{l}$ in $\mathrm{\Pi}$, suppose that $v_{1},\dots,v_{n}$ are the free variables in it, we add the formula
\[
\mathrm{FUL}_{m}v_{1}v_{2}\cdots v_{n-1}v_{n}v_{n}\cdots v_{n}
\]
to the body and obtain a new rule
\[
\beta\leftarrow\alpha_{1},\dots,\alpha_{l},\mathrm{FUL}_{m}v_{1}v_{2}\cdots v_{n-1}v_{n}\cdots v_{n}.
\]
For each new rule, we
\begin{itemize}
	
	\item replace $x=y$ by $i_{x}=i_{x\barwedge y},i_{x\barwedge y}=i_{y}$
	(reason: $d(x)=d(x\barwedge y)=d(y)$), for constant $\mathbf{root}$,
	we replace $i_{\mathbf{root}}$ by constant $\mathbf{0}$, and replace $i_{\mathbf{root}\barwedge x}$ also
	by $\mathbf{0}$, since $\mathbf{root}\barwedge a=\mathbf{root}$ for any node $a$;
	
	\item replace $Exy$ by $i_{x}=i_{x\barwedge y},\mathrm{SUCC}i_{x\barwedge y}i_{y}$
	(reason: $d(y)=d(x\barwedge y)+1=d(x)+1$);
	
	\item replace $x\neq y$ by $R_{\neq}i_{x}i_{x\barwedge y}i_{y}$ (reason:
	$R_{\neq}$ is the characteristic relation of $\neq$);
	
	\item replace $\neg Exy$ by $R_{\neg e}i_{x}i_{x\barwedge y}i_{y}$ (reason:
	$R_{\neg e}$ is the characteristic relation of $\neg E$);
	
	\item replace $P\bar{x}$ by $P^{*}(\bar{x})^{*}$, where $P$ is in  $\{R_{1},\dots,R_{k},\mathrm{FUL}_{m}\}$, or an intensional relation symbol;
	
	\item replace $\neg R\bar{x}$ by $(\neg R)^{*}(\bar{x})^{*}$,
	where $R$ is a symbol in $\{R_{1},\dots,R_{k}\}$;
	
	\item replace $\forall y_{1}\cdots\forall y_{t}Py_{1}\cdots y_{t}z_{1}\cdots z_{s}$ by
	\[
	\Psi_{P}=\left(\begin{array}{l}
		\mathrm{FUL}_{m}^{*}(z_{1}z_{2}\cdots z_{s-1}z_{s}\cdots z_{s})^{*}\wedge\\
		\forall(y_{1}\cdots y_{t})^{*}\forall i_{y_{1}\barwedge z_{1}}\cdots\forall i_{y_{1}\barwedge z_{s}}\forall i_{y_{2}\barwedge z_{1}}\cdots\forall i_{y_{t}\barwedge z_{s-1}}\forall i_{y_{t}\barwedge z_{s}}\\
		\bigr(\mathrm{FUL}_{m}^{*}(y_{1}\cdots y_{t}z_{1}\cdots z_{s}z_{s}\cdots z_{s})^{*} \rightarrow
		P^{*}(y_{1}\cdots y_{t}z_{1}\cdots z_{s})^{*}\bigl)
	\end{array}\right)
	\]
	where $P$ is an intensional relation symbol.
\end{itemize}

By adding $\mathrm{FUL}_m$ to each rule of $\mathrm{\Pi}$ and replacing it with $\mathrm{FUL}_m^{*}$ in the translation, we can restrict $\mathrm{\Pi}'$ to characteristic tuples, i.e., only invariant relations are considered. $\mathrm{\Pi}'$ is not a Datalog$^r$ program because of $\Psi_P$. Note that $\Psi_P$ is equivalent to the Datalog$^r$ formula $(\mathrm{\Pi}_{1},Q_{2})\bar{t}$, where
\[
\begin{aligned}
	\mathrm{\Pi}_{1}: Q(y_{1}\cdots y_{t}z_{1}\cdots z_{s})^{*} & \leftarrow \neg \mathrm{FUL}_{m}^{*}(y_{1}\cdots y_{t}z_{1}\cdots z_{s}z_{s}\cdots z_{s})^{*};\\
	Q(y_{1}\cdots y_{t}z_{1}\cdots z_{s})^{*} & \leftarrow  P^{*}(y_{1}\cdots y_{t}z_{1}\cdots z_{s})^{*};\\
	Q_{1}(z_{1}\cdots z_{s})^{*} & \leftarrow  \forall(y_{1}\cdots y_{t})^{*}\forall i_{y_{1}\barwedge z_{1}}\cdots\forall i_{y_{1}\barwedge z_{s}}\\
	&\qquad\! \forall i_{y_{2}\barwedge z_{1}}\cdots\forall i_{y_{t}\barwedge z_{s-1}}\forall i_{y_{t}\barwedge z_{s}}Q(y_{1}\cdots y_{t}z_{1}\cdots z_{s})^{*};\\
	Q_{2}(z_{1}\cdots z_{s})^{*} & \leftarrow  Q_{1}(z_{1}\cdots z_{s})^{*},\mathrm{FUL}_{m}^{*}(z_{1}z_{2}\cdots z_{s-1}z_{s}\cdots z_{s})^{*}.
\end{aligned}
\]
The Datalog$^r$ program $\mathrm{\Pi}^{*}$ can be obtained by adding $\mathrm{\Pi}_{1}$ to $\mathrm{\Pi}'$ and changing $\Psi_{P}$ to $Q_{2}(z_1\cdots z_s)^*$.

\begin{remark}
	We cannot replace $\forall y_{1}\cdots\forall y_{t}Py_{1}\cdots y_{t}z_{1}\cdots z_{s}$ directly by 
	\[
	\forall(y_{1}\cdots y_{t})^{*}\forall i_{y_{1}\barwedge z_{1}}\cdots\forall i_{y_{1}\barwedge z_{s}}\forall i_{y_{2}\barwedge z_{1}}\cdots\forall i_{y_{t}\barwedge z_{s-1}}\forall i_{y_{t}\barwedge z_{s}} P^{*}(y_{1}\cdots y_{t}z_{1}\cdots z_{s})^{*}
	\]
	since there may be $\mathbf{T}\in\mathcal{T}$, $\bar{a}\in \mathbf{T}$,
	and an invariant relation $P$ such that
	\[
	\begin{aligned}
		\mathbf{T} & \vDash \forall y_{1}\cdots\forall y_{t}P y_{1}\cdots y_{t}z_{1}\cdots z_{s}[\bar{a}], \text{and}\\
		S_{\mathbf{T}} & \nvDash \forall(y_{1}\cdots y_{t})^{*}\forall i_{y_{1}\barwedge z_{1}}\cdots\forall i_{y_{1}\barwedge z_{s}}\forall i_{y_{2}\barwedge z_{1}}\cdots\forall i_{y_{t}\barwedge z_{s-1}}\forall i_{y_{t}\barwedge z_{s}}\\
		& \quad P^{*}(y_{1}\cdots y_{t}z_{1}\cdots z_{s})^{*}[(\bar{a})^{*}].
	\end{aligned}
	\]
	For example, let $\mathbf{T}$ be the structure with the perfect binary tree of Fig.~\ref{fig-binarytree} and relation
	\[
	 P=\{(\mathbf{root},\mathbf{root}),(\mathbf{root},a),(\mathbf{root},b),(\mathbf{root},c),(\mathbf{root},d),(\mathbf{root},e),(\mathbf{root},f)\}.
	\]
	Obviously, we have $\mathbf{T} \vDash \forall y P(\mathbf{root},y)$. But $S_{\mathbf{T}} \nvDash P^{*}(\mathbf{0}, 1,1) $ since $(\mathbf{0}, 1,1)$ is not the characteristic tuple of any tuple of nodes in $\mathbf{T}$.
	This problem can be solved by changing $\forall y_{1}\cdots\forall y_{t}Py_{1}\cdots y_{t}z_{1}\cdots z_{s}$ to the equivalent formula
	\[
	\begin{array}{l}
		\mathrm{FUL}_{m}z_{1}\cdots z_{s-1}z_{s}z_{s}\cdots z_{s}\,\wedge\\
		\forall y_{1}\cdots\forall y_{t}(\mathrm{FUL}_{m}y_{1}\cdots y_{t}z_{1}\cdots z_{s}z_{s}\cdots z_{s}\rightarrow Py_{1}\cdots y_{t}z_{1}\cdots z_{s}).
	\end{array}
	\]
	We replace $\mathrm{FUL}_{m}$ and $P$ by their characteristic relations $\mathrm{FUL}_{m}^*$ and $P^*$, respectively, to obtain $\Psi_{P}$. This guarantees that only characteristic tuples are considered.
\end{remark}

\begin{example}
The following Datalog$^r$ program $\mathrm{\Pi}$ computes the transitive closure $R$ of edges $E$
	\[
	\begin{aligned}
		\mathrm{\Pi}: R x_1x_2 & \leftarrow E x_1x_2;\\
		 R x_1x_3 & \leftarrow  R x_1x_2, E x_2x_3.
	\end{aligned}
	\]
	The corresponding Datalog$^r$ program $\mathrm{\Pi}^*$ below computes the characteristic relation $R^*$ of $R$.
	\[
	\begin{aligned}
		\mathrm{\Pi}^*: R^* i_{x_1} i_{x_1\barwedge x_2} i_{x_2} & \leftarrow i_{x_1}=i_{x_1\barwedge x_2},\mathrm{SUCC}i_{x_1\barwedge x_2}i_{x_2}, \\
		& \qquad\! \mathrm{FUL}_3^* i_{x_1} i_{x_1\barwedge x_2} i_{x_1\barwedge x_2} i_{x_2} i_{x_2} i_{x_2};\\
		R^* i_{x_1} i_{x_1\barwedge x_3} i_{x_3} & \leftarrow  R^* i_{x_1} i_{x_1\barwedge x_2} i_{x_2}, 
		 i_{x_2}=i_{x_2\barwedge x_3},\mathrm{SUCC}i_{x_2\barwedge x_3}i_{x_3}, \\
		 & \qquad\! \mathrm{FUL}_3^* i_{x_1} i_{x_1\barwedge x_2} i_{x_1\barwedge x_3} i_{x_2} i_{x_2\barwedge x_3} i_{x_3}.
	\end{aligned}
	\]
\end{example}

\begin{lemma}
	\label{lem:lem4} Given $\psi_{P}=\forall y_{1}\cdots\forall y_{t}Py_{1}\cdots y_{t}z_{1}\cdots z_{s}$, a structure $\mathbf{T}\in \mathcal{T}$, let $\Psi_{P}$ be defined as above, and $Q_{1}=\{\bar{a}\mid \mathbf{T}\models\psi_{P}[\bar{a}]\}$,
	$Q_{2}=\{\bar{e}\mid S_{\mathbf{T}}\models\Psi_{P}[\bar{e}]\}$.
	If $P$ is an invariant relation on $\mathbf{T}$, then $(Q_{1})^{*}=Q_{2}$.
\end{lemma}
\begin{proof}
	Because $P$ is an invariant relation, by Lemma~\ref{lem:lem3} and the definition of $Q_{1}$, we know that $Q_{1}$ is also an invariant relation. We first show that $(Q_{1})^{*}\subseteq Q_{2}$. Suppose that
	$\bar{e}\in(Q_{1})^{*}$ for some $\bar{e}\in S_{\mathbf{T}}$, there must
	exist a tuple $\bar{a}$ from $\mathbf{T}$ such that $\bar{a}\in Q_{1}$, $(\bar{a})^{*}=\bar{e}$,
	and $\bar{b}\bar{a}\in P$ for all tuples $\bar{b}$ of $\mathbf{T}$,
	i.e.,
	\[
	\begin{array}{ll}
		\mathbf{T}\models & \Bigl(\mathrm{FUL}_{m}z_{1}\cdots z_{s-1}z_{s}z_{s}\cdots z_{s}\wedge\\
		& \forall y_{1}\cdots\forall y_{t}(\mathrm{FUL}_{m}y_{1}\cdots y_{t}z_{1}\cdots z_{s}z_{s}\cdots z_{s}\rightarrow Py_{1}\cdots y_{t}z_{1}\cdots z_{s})\Bigr)[\bar{a}].
	\end{array}
	\]
	By the definition of $\Psi_P$  we see that $S_{\mathbf{T}}\models\Psi_{P}[(\bar{a})^{*}]$, which implies $\bar{e}\in Q_{2}$.
	
	To prove $Q_{2}\subseteq(Q_{1})^{*}$, consider an arbitrary tuple $\bar{e}\in S_{\mathbf{T}}$ such that $\bar{e}\in Q_{2}$.
	By the definition of $Q_{2}$ and $\Psi_{P}$, we have $S_{\mathbf{T}}\models \mathrm{FUL}_m(z_1\cdots z_sz_s\cdots z_s)^*[\bar{e}]$, so there exists a tuple $\bar{a}$ of $\mathbf{T}$
	such that $\bar{e}=(\bar{a})^{*}$.  On the contrary, assume $\bar{a}\notin Q_{1}$, then there is a tuple $\bar{b}$ such that $\bar{b}\bar{a}\notin P$. Because $P$ is an invariant relation, for any tuples $\bar{b}'$ and $\bar{a}'$, if $(\bar{b}\bar{a})^{*}=(\bar{b}'\bar{a}')^{*}$ then $\bar{b}'\bar{a}'\notin P$. Combing that $P^{*}$ is the characteristic relation of $P$ we conclude that
	\begin{align}
		S_{\mathbf{T}} & \nvDash P^{*}(y_{1}\cdots y_{t}z_{1}\cdots z_{s})^{*}[(\bar{b}\bar{a})^{*}]\label{eq:lem1}, \text{and}\\
		S_{\mathbf{T}} & \vDash \mathrm{FUL}_{m}^{*}(y_{1}\cdots y_{t}z_{1}\cdots z_{s}z_{s}\cdots z_{s})^{*}[(\bar{b}\bar{a})^{*}].\label{eq:lem2}
	\end{align}
	(\ref{eq:lem1}) and (\ref{eq:lem2}) give $S_{\mathbf{T}}\nvDash\Psi_{P}[(\bar{a})^{*}]$. Hence, $\bar{e}\notin Q_{2}$, contrary to the assumption that $\bar{e}\in Q_{2}$. Therefore, $\bar{a}$ must be in $Q_1$, which implies $\bar{e}\in (Q_1)^*$. \qed
\end{proof}

Let $P$ be an intensional relation symbol in $\mathrm{\Pi}$, and $\mathbf{T}$ a structure in $\mathcal{T}$. We use $P_{(n)}\,(n>0)$ to denote the relation obtained in the $n$-th evaluation of $\mathrm{\Pi}$ on $\mathbf{T}$ for $P$, and $P^{\mathbf{T}[\mathrm{\Pi}]}$ to denote the relation obtained by applying $\mathrm{\Pi}$ on $\mathbf{T}$ for $P$, i.e., the fixpoint of the sequence $P_{(0)},P_{(1)},P_{(2)},\dots$

\begin{proposition}\label{prop:treeEqualstring} 
	For any intensional relation symbol $P$ in $\mathrm{\Pi}$ and any $\mathbf{T}\in\mathcal{T}$, $P^{\mathbf{T}[\mathrm{\Pi}]}$ is an invariant relation on $\mathbf{T}$  and $(P^{\mathbf{T}[\mathrm{\Pi}]})^{*}=(P^{*})^{S_{\mathbf{T}}[\mathrm{\Pi}^{*}]}$.  Moreover, if $P$
	is a 0-ary intensional relation symbol, then $\mathbf{T}\models(\mathrm{\Pi},P)$ iff $S_{\mathbf{T}}\models(\mathrm{\Pi}^{*},P^{*})$.
\end{proposition}
\begin{proof}
		We first show that if $P$ is an intensional relation symbol in $\mathrm{\Pi}$ and
	$\mathbf{T}$ is a structure in $\mathcal{T}$, then $P^{\mathbf{T}[\mathrm{\Pi}]}$
	is an invariant relation on $\mathbf{T}$. Let
	$P^{1},\dots,P^{m'}$ be all intensional relation symbols in $\mathrm{\Pi}$.
	Consider the following formula constructed for each $P^{i}$
	\[
	\begin{array}{r}
		\phi_{P^{i}}(\bar{x}_{P^{i}})=\bigvee\{\exists\bar{v}(\alpha_{1}\wedge\cdots\wedge\alpha_{l})\mid P^{i}\bar{x}_{P^{i}}\leftarrow\alpha_{1},\dots,\alpha_{l}\in\mathrm{\Pi} \text{ and }\bar{v}\text{ are the}\\
		\text{free variables in }\alpha_{1}\wedge\cdots\wedge\alpha_{l} \text{ that are different from }\bar{x}_{P^{i}}\}.
	\end{array}
	\]
	If the relation defined by each $\alpha_{s}$ is an invariant relation,
	then by Lemmas~\ref{lem:lem1}, \ref{lem:lem2} and \ref{lem:lem3},
	we know that the relation defined by $\phi_{P^{i}}$
	is also an invariant relation. Each $\alpha_{s}$ is either an atomic (or negated atomic) formula with relation symbol from $\{=,E,R_{1},\dots,R_{k}\}$, where the relations defined by them are all invariant relations, or an atomic formula $P^{j}\bar{x}$ or a formula $\forall\bar{y}P^{j}\bar{y}\bar{z}\, (1\leq j\leq m')$.
	
	When computing  the fixpoint of $P^{1},\dots,P^{m'}$, we set
	$P_{(0)}^{i}=\emptyset\,(1\leq i\leq m')$, where $\emptyset$ is
	an invariant relation. By Lemma~\ref{lem:lem3} we know that if $P^{j}$
	is an invariant relation then the relation defined by $\forall\bar{y}P^{j}\bar{y}\bar{z}$
	is also an invariant relation. We proceed by induction on $n$. Suppose that $P_{(n)}^{1},\dots,P_{(n)}^{m'}$ are invariant relations, then each
	\begin{align*}
		P_{(n+1)}^{i} & =\{\bar{a}\mid(\mathbf{T},P_{(n)}^{1},\dots,P_{(n)}^{m'})\models\phi_{P^{i}}(\bar{x}_{P^{i}})[\bar{a}]\}, \text{or}\\
		P_{(n+1)}^{i} & =\{\emptyset\mid(\mathbf{T},P_{(n)}^{1},\dots,P_{(n)}^{m'})\models\phi_{P^{i}}\}
	\end{align*}
	is also an invariant relation. Therefore, the fixpoints $P_{(\infty)}^{1},\dots,P_{(\infty)}^{m'}$
	are invariant relations, i.e., $P^{\mathbf{T}[\mathrm{\Pi}]}$ is an invariant
	relation on $\mathbf{T}$.
	
	Next we shall show that $(P^{\mathbf{T}[\mathrm{\Pi}]})^{*}=(P^{*})^{S_{\mathbf{T}}[\mathrm{\Pi}^{*}]}$.
	It suffices to prove that $(P_{(n)}^{i})^{*}=(P^{i})_{(n)}^{*}\,(1\leq i\leq m')$
	for each $n\geq0$. The proof is by induction on $n$.
	
	\vskip 2mm \textbf{Basis:} If $n=0$, then $P_{(0)}^{i}=\emptyset$, $(P^{i})_{(0)}^{*}=\emptyset\,(1\leq i\leq m')$.
	We have $(P_{(0)}^{i})^{*}=(P^{i})_{(0)}^{*}\,(1\leq i\leq m')$.
	
	\vskip 2mm \textbf{Inductive step:} Assuming $(P_{(k)}^{i})^{*}=(P^{i})_{(k)}^{*}$ $(1\leq i\leq m')$, we show that $(P_{(k+1)}^{i})^{*}=(P^{i})_{(k+1)}^{*}$ $(1\leq i\leq m')$.
	The case where $P^i$ is 0-ary is trivial, in the following we only consider the relation $P^i$ of no 0-ary.
	
	To prove $(P_{(k+1)}^{i})^{*}\subseteq(P^{i})_{(k+1)}^{*}$, suppose
	$\bar{e}\in(P_{(k+1)}^{i})^{*}$ for some $\bar{e}\in S_{\mathbf{T}}$. There must be a tuple $\bar{a}$ of $\mathbf{T}$ such that $\bar{a}\in P_{(k+1)}^{i}$
	and $\bar{e}=(\bar{a})^{*}$. By the semantics of Datalog$^r$
	we know that
	\[
	P_{(k+1)}^{i}=\{\bar{a}\mid(\mathbf{T},P_{(k)}^{1},\dots,P_{(k)}^{m'})\models\phi_{P^{i}}(\bar{x}_{P^{i}})[\bar{a}]\}.
	\]
	By the definition of $\phi_{P^i}$, there is a rule $P^{i}\bar{x}_{P^{i}}\leftarrow\alpha_{1},\dots,\alpha_{l}$
	in $\mathrm{\Pi}$ such that
	\[
	\langle \mathbf{T},P_{(k)}^{1},\dots,P_{(k)}^{m'} \rangle \models\exists\bar{v}(\alpha_{1}\wedge\cdots\wedge\alpha_{l})[\bar{a}].
	\]
	Thus, there exists some $\bar{b}$ such that
	\[
	\langle \mathbf{T},P_{(k)}^{1},\dots,P_{(k)}^{m'} \rangle \models(\alpha_{1}\wedge\cdots\wedge\alpha_{l})[\bar{a}\bar{b}].
	\]
	Because $P^{i}\bar{x}_{P^{i}}\leftarrow\alpha_{1},\dots,\alpha_{l}$
	is a rule of $\mathrm{\Pi}$, we can infer that
	\[
	(P^{i})^{*}(\bar{x}_{P^{i}})^{*}\leftarrow\alpha_{1}',\dots,\alpha_{l}',\mathrm{FUL}_{m}^{*}(\bar{x}_{P^{i}}\bar{v}\tilde{v}')^{*}
	\]
	is a rule of $\mathrm{\Pi}^{*}$, where $\alpha_{1}',\dots,\alpha_{l}'$ and $\mathrm{FUL}_{m}^{*}(\bar{x}_{P^{i}}\bar{v}\tilde{v}')^{*}$ are
	obtained by replacing $\alpha_{1},\dots,\alpha_{l},\mathrm{FUL}_{m}$ with the corresponding
	formulas respectively in the construction of $\mathrm{\Pi}^*$.
	Note that we replace $\forall\bar{y}P\bar{y}\bar{z}$ by $\Psi_{P}$, and by Lemma~\ref{lem:lem4} the relation defined by $\Psi_{P}$ is the characteristic relation of that defined by $\forall\bar{y}P\bar{y}\bar{z}$. By the definition of $S_{\mathbf{T}}$ and the induction hypothesis $(P_{(k)}^{i})^{*}=(P^{i})_{(k)}^{*}\,(1\leq i\leq m')$ we deduce that
	\begin{align*}
		\langle S_{\mathbf{T}},(P^{1})_{(k)}^{*},\dots,(P^{m'})_{(k)}^{*}\rangle & \models(\alpha_{1}'\wedge\cdots\wedge\alpha_{l}'\wedge \mathrm{FUL}_{m}^{*}(\bar{x}_{P^{i}}\bar{v}\tilde{v}')^{*})[(\bar{a}\bar{b})^{*}],\text{i.e.,}\\
		\langle S_{\mathbf{T}},(P^{1})_{(k)}^{*},\dots,(P^{m'})_{(k)}^{*}\rangle & \models\exists\bar{u}(\alpha_{1}'\wedge\cdots\wedge\alpha_{l}'\wedge \mathrm{FUL}_{m}^{*}(\bar{x}_{P^{i}}\bar{v}\tilde{v}')^{*})[(\bar{a})^{*}]
	\end{align*}
	where $\bar{u}$ are the free variables in $\alpha_{1}'\wedge\cdots\wedge\alpha_{l}'\wedge \mathrm{FUL}_{m}^{*}(\bar{x}_{P^{i}}\bar{v}\tilde{v}')^{*}$
	that are different from $(\bar{x}_{P^{i}})^{*}$. Combining $\bar{e}=(\bar{a})^{*}$
	we obtain $\bar{e}\in(P^{i})_{(k+1)}^{*}$.
	
	To prove $(P^{i})_{(k+1)}^{*}\subseteq(P_{(k+1)}^{i})^{*}$, suppose $\bar{e}\in(P^{i})_{(k+1)}^{*}$ for some $\bar{e}\in S_{\mathbf{T}}$.
	There must exist a rule
	\begin{equation}
		(P^{i})^{*}(\bar{x}_{P^{i}})^{*}\leftarrow\alpha_{1}',\dots,\alpha_{l}',\mathrm{FUL}_{m}^{*}(\bar{x}_{P^{i}}\bar{v}\tilde{v}')^{*}\label{eq:rule1}
	\end{equation}
	in $\mathrm{\Pi}^{*}$ such that
	\[
	\langle S_{\mathbf{T}},(P^{1})_{(k)}^{*},\dots,(P^{m'})_{(k)}^{*} \rangle \models\exists\bar{u}(\alpha_{1}'\wedge\cdots\wedge\alpha_{l}'\wedge \mathrm{FUL}_{m}^{*}(\bar{x}_{P^{i}}\bar{v}\tilde{v}')^{*})[\bar{e}].
	\]
	Hence there exists a tuple $\bar{f}\in S_{\mathbf{T}}$ such that
	\[
	(S_{\mathbf{T}},(P^{1})_{(k)}^{*},\dots,(P^{m'})_{(k)}^{*})\models(\alpha_{1}'\wedge\cdots\wedge\alpha_{l}'\wedge \mathrm{FUL}_{m}^{*}(\bar{x}_{P^{i}}\bar{v}\tilde{v}')^{*})[\bar{e}\bar{f}].
	\]
	The formula $\mathrm{FUL}_{m}^{*}(\bar{x}_{P^{i}}\bar{v}\tilde{v}')^{*}$ guarantees
	that $(\bar{a}\bar{b})^{*}=\bar{e}\bar{f}$ and $(\bar{a})^{*}=\bar{e}$
	for some tuple $\bar{a}\bar{b}$ of $\mathbf{T}$. Because $P_{(k)}^{1},\dots,P_{(k)}^{m'}$ are invariant relations, by the induction hypothesis $(P_{(k)}^{i})^{*}=(P^{i})_{(k)}^{*}$ $(1\leq i\leq m')$
	we know that
	\[
	\langle \mathbf{T},P_{(k)}^{1},\dots,P_{(k)}^{m'} \rangle \models(\alpha_{1}\wedge\cdots\wedge\alpha_{l})[\bar{a}\bar{b}]
	\]
	where $\alpha_{1},\dots,\alpha_{l}$ occur in the rule
	$P^{i}\bar{x}_{P^{i}}\leftarrow\alpha_{1},\dots,\alpha_{l}$ that is the original of (\ref{eq:rule1}) in $\mathrm{\Pi}$. Hence $\bar{a}\in P_{(k+1)}^{i}$, which implies $\bar{e}\in(P_{(k+1)}^{i})^{*}$. This completes the proof. \qed
\end{proof}

\subsection{Nondefinability results for Datalog$^r$}\label{sec-nondef}

The complexity class EXPTIME contains the decision problems decidable by a deterministic Turing machine in $O(2^{n^{c}})$ time. By the time hierarchy theorem, we know that PTIME is a proper subset of EXPTIME. In this section, for every class of structures $\mathcal{K}\in$ EXPTIME, we construct a class $\mathcal{K}'$ of structures that is in PTIME and closed under substructures, and show that if $\mathcal{K}'$ is definable by a Datalog$^r$ formula, then $\mathcal{K}$ is in P, which is impossible.

Let $c$ be a constant, and $\mathbf{A}=\langle \{0,\dots, h-1\},$ $ R^A_1,\dots, R^A_k\rangle$ a $\tau$-structure. The \textit{trivial extension} $\mathbf{A}^+=\langle \{0,\dots,h-1,h,\dots, h+h^c-1\},  R^A_1,\dots, R^A_k\rangle$ of $\mathbf{A}$ is a $\tau$-structure obtained by adding $h^c$ dummy elements to the domain of $\mathbf{A}$ and keeping all other relations unchanged.

For technical reasons we introduce a new unary relation symbol $U$ and let $\tau_U=\tau\cup\{U\}$,
$\tau'_U=\tau\cup\{\mathbf{root}, E, U\}$. From now on when we speak of a $\tau'_U$-structure 
\[ \mathbf{G}=\langle V, \mathbf{root}, E, U, R_1,\cdots, R_k\rangle\]
we assume that
\begin{description}
	\item[(1)] $\langle V,E,\mathbf{root} \rangle$ is a directed acyclic graph and the nodes reachable from $\mathbf{root}$ form a binary tree, and
	
	\item[(2)] all relations $U, R_1,\cdots, R_k$ are saturated relations restricted on $T(\mathbf{G})$, which is the largest perfect binary subtree of $\mathbf{G}$ with $\mathbf{root}$ as the root.
\end{description}
It is easy to check that if a $\tau'_U$-structure $\mathbf{G}$ satisfies the aforementioned two conditions, then all its substructures also satisfy the two conditions.

\begin{definition}\label{fz6} 
	Let $\mathcal{K}$ be a class of $\tau$-structures. Define a class $\mathcal{K}'$ of $\tau'_U$-structures such that, for any $\mathbf{G}=\langle V,\mathbf{root},E,U,R_{1},\dots,R_{k}\rangle$, let $h$ be the largest number where all nodes in the first $h$ levels of $T(\mathbf{G})$ are marked by $U$, $\mathbf{G}\in \mathcal{K}'$ iff the following Condition~(1) or Condition~(2) holds.
	
	\begin{description}
		\item[Condition (1)]\quad
		\begin{description}
			\item[(a)] The depth of $T(\mathbf{G})$ is $h+h^c$.
			\item[(b)] The relations $R_{1},\dots,R_{k}$ do not hold on any tuple that contains a node in the last $h^c$ consecutive levels of $T(\mathbf{G})$.
			\item[(c)] $C^{-1}(T(\mathbf{G}))$ is the trivial extension of $C^{-1}(T_h(\mathbf{G}))$, where $T_h(\mathbf{G})$ is the subtree of $T(\mathbf{G})$ by restricting to the first $h$ levels.
			\item[(d)] $C^{-1}(T_h(\mathbf{G}))\in \mathcal{K}$ when ignoring the relation $U$.
		\end{description}
		
		\item[Condition (2)]\quad
		\begin{description}
			\item[(a)] The depth of $T(\mathbf{G})$ is strictly less than $h+h^c$.
		\end{description}
	\end{description}
\end{definition}

\begin{proposition}\label{pfz1}
	Let $\mathcal{K}$ be an arbitrary class of $\tau$-structures decidable in $2^{n^c}$ time, where $n$ is the cardinality of the structure's domain, and $\mathcal{K}'$ defined as above. Then
	\begin{description}
		\item[(i)] $\mathcal{K}'$ is closed under substructures;
		\item[(ii)] $\mathcal{K}'$ is decidable in $\mathrm{PTIME}$.
	\end{description}
\end{proposition}
\begin{proof}
		To prove (i), suppose that $\mathbf{G}$ is
	a $\tau'_U$-structure in $\mathcal{K}'$, then it satisfies either Condition~(1) or Condition~(2) in Definition~\ref{fz6}. Let $\mathbf{H}$ be an arbitrary substructure of
	$\mathbf{G}$. If $\mathbf{G}$ satisfies Condition~(2), then $\mathbf{H}$ also satisfies Condition~(2), and is in $\mathcal{K}'$. Suppose that $\mathbf{G}$ satisfies Condition~(1), then the perfect binary tree $T(\mathbf{H})$ either equals $T(\mathbf{G})$, which implies $\mathbf{H}$ satisfies Condition~(1), or the depth of $T(\mathbf{H})$ is less than that of $T(\mathbf{G})$, which implies $\mathbf{H}$ satisfies Condition~(2). Altogether, $\mathbf{H}\in \mathcal{K}'$.
	
	To prove (ii), let $\mathbf{G}$ be an arbitrary $\tau'_U$-structure, we just need to do the following steps to check whether $\mathbf{G}\in \mathcal{K}'$:
	\begin{description}
		\item[(1)] Check that $\langle V, E\rangle$ is a directed acyclic graph.
		
		\item[(2)] Check that all nodes reachable from $\mathbf{root}$ form a binary tree.
		
		\item[(3)] Compute $T(\mathbf{G})$, the largest perfect binary subtree with $\mathbf{root}$ as root.
		
		\item[(4)] Check that $U,R_1,\dots,R_k$ are saturated relations on $T(\mathbf{G})$.
		
		\item[(5)] Compute the largest number $h$ such that all nodes in the first $h$ levels of $T(\mathbf{G})$ have property $U$.
		
		\item[(6)] Check whether the depth of $T(\mathbf{G})$ is less than $h+h^c$.
		
		\item[(7)] If the depth of $T(\mathbf{G})$ is $h+h^c$, then check whether (b), (c) and (d) in Condition~(1) of Definition~\ref{fz6} hold.
	\end{description}
	Note that $C^{-1}(T_h(\mathbf{G}))$ has $h$ elements and $\mathcal{K}$ is decidable in $2^{n^c}$ time, the statement (d) in Condition~(1) of Definition~\ref{fz6} can be verified in polynomial time since if the depth of $T(\mathbf{G})$ is $h+h^c$ then the input size is at least $2^{h+h^c}$. \qed
\end{proof}

Let $\mathcal{K}$ and $\mathcal{K}'$ be defined as in Proposition~\ref{pfz1}.
For an arbitrary $\tau$-structure $\mathbf{A}$, let $\mathbf{A}_U$ be the $\tau_U$-structure obtained by marking every element in $\mathbf{A}$ by $U$, $\mathbf{A}^+_U$ the trivial extension of $\mathbf{A}_U$ by adding $|A|^c$ elements, and $\mathbf{T}$ the $\tau'_U$-structure such that $C^{-1}(\mathbf{T})=\mathbf{A}^+_U$. If $\mathcal{K}'$ is axiomatizable by a Datalog$^r$ formula $(\mathrm{\Pi},Q)$, then
\begin{equation} \label{eqfz3} 
	\mathbf{A}\in \mathcal{K}\mbox{ \ \ iff \ \ }\mathbf{T}\in \mathcal{K}'\mbox{\ \  iff\ \  }\mathbf{T}\models (\mathrm{\Pi},Q).
\end{equation}
Define the vocabulary
\[
\sigma_U=\{\mathbf{0},\mathrm{SUCC},R_{\neq},R_{\neg e},\mathrm{FUL}_{m}^{*},U^{*},R_{1}^{*},\dots,R_{k}^{*},(\neg U)^{*},(\neg R_{1})^{*},\dots,(\neg R_{k})^{*}\}.
\]
By Definition~\ref{fz4}, we can compute $\mathbf{T}$'s characteristic structure that is
a $\sigma_U$-structure
\[
\begin{array}{lr}
	S_{\mathbf{T}}= & \bigl\langle\{0,1,\dots,|A|+|A|^c-1\},\mathbf{0},\mathrm{SUCC},R_{\neq},R_{\neg e},\mathrm{FUL}_{m}^{*},U^{*},\\
	& R_{1}^{*},\dots,R_{k}^{*},(\neg U)^{*},(\neg R_{1})^{*},\dots,(\neg R_{k})^{*}\bigr\rangle
\end{array}
\]
where $\mathbf{0}$ is interpreted by 0, $\mathrm{SUCC}$ is the successor relation on the domain
and $R_{\neq},R_{\neg e}$, $\mathrm{FUL}_{m}^{*},U^{*}$, $R_{1}^{*},\dots,R_{k}^{*},(\neg U)^{*},(\neg R_{1})^{*}, \dots, (\neg R_{k})^{*}$ are the characteristic relations of $\neq,(\neg E)$, $\mathrm{FUL}_{m},U$, $R_{1}, \dots$, $R_{k}$, $\neg U,\neg R_{1},\dots,\neg R_{k}$, respectively.
By Proposition~\ref{prop:treeEqualstring}, we know there is a Datalog$^r$ formula $(\mathrm{\Pi}^*, Q^*)$ such that $\mathbf{T}\models (\mathrm{\Pi},Q)$ iff $S_{\mathbf{T}}\models (\mathrm{\Pi}^*, Q^*)$. Combining (\ref{eqfz3}), we have $\mathbf{A}\in\mathcal{K}$ iff $S_{\mathbf{T}}\models (\mathrm{\Pi}^*,Q^*)$. 

\begin{lemma}\label{lem-logspace-comp}
	$S_{\mathbf{T}}$ is logspace computable from $\mathbf{A}$.
\end{lemma}
\begin{proof}
Given a $\tau$-structure $\mathbf{A}=\langle \{0,1,2,\dots,h-1\},R_{1},\dots,R_{k}\rangle$ as the input, we shall compute the corresponding $\sigma_U$-structure
	\[
	\begin{array}{lr}
		S_{\mathbf{T}}= & \bigl\langle\{0,1,\dots,h+h^{c}-1\},\mathbf{0},\mathrm{SUCC},R_{\neq},R_{\neg e},\mathrm{FUL}_{m}^{*},U^{*},\,\\
		& R_{1}^{*},\dots,R_{k}^{*},(\neg U)^{*},(\neg R_{1})^{*},\dots,(\neg R_{k})^{*}\bigr\rangle.
	\end{array}
	\]
	Obviously, $\{0,1,\dots,h+h^{c}-1\},\mathbf{0},\mathrm{SUCC},U^{*},(\neg U)^{*}$ are computable in logspace, where $U^{*}= \{0,1,\dots,h-1\}$ and $(\neg U)^{*}=\{h,h+1,\dots,h+h^{c}-1\}$.
	
	Consider the relation $R_{\neq}$, a tuple $(e_{1},e_{2},e_{3})\in R_{\neq}$ iff $(e_{1},e_{2},e_{3})$ is the characteristic tuple $(d(a),d(a\barwedge b),d(b))$ of a tuple
	$(a,b)$ of nodes of $\mathbf{T}$ such that $C^{-1}(\mathbf{T})=\mathbf{A}^+$ and $a\neq b$.  If $d(a)=d(b)$, then $a,b$ are in the same level, hence $d(a\barwedge b)<d(a)$.
	Otherwise, $d(a\barwedge b)\leq\min\{d(a),d(b)\}$. 
	
	\begin{algorithm}[H]
		\SetKwInOut{Input}{Input}
		\SetKwInOut{Output}{Output}
		\Input{A tuple $(e_{1},e_{2},e_{3})$.}
		\Output{Accept if $(e_{1},e_{2},e_{3})\in R_{\neq}$, and reject otherwise.}
		\eIf{$e_{1}=e_{3}$}{
			\leIf{$e_{2}<e_{1}$}{accept} {reject}
		}{
			\leIf{$e_{2}\leq\min\{e_{1},e_{3}\}$}{accept} {reject}
		}
		\caption{\label{algo-relation-euqa}The procedure to decide $R_{\neq}$.}
	\end{algorithm}
	
	For the relation $R_{\neg e}$, a tuple $(e_{1},e_{2},e_{3})\in R_{\neg e}$
	iff $(e_{1},e_{2},e_{3})$ is the characteristic tuple $(d(a),d(a\barwedge b),d(b))$
	of a tuple $(a,b)$ of nodes of $\mathbf{T}$ such that $C^{-1}(\mathbf{T})=\mathbf{A}^+$ and $\neg Eab$.
	Since $Eab$ iff $d(a)+1=d(b)$ and $d(a\barwedge b)=d(a)$, thus,
	if $d(a)+1=d(b)$ then $\neg Eab$ iff $d(a\barwedge b)<d(a)$, and
	if $d(b)\leq d(a)$ or $d(a)<d(b)-1$ then $\neg Eab$.
	
	\begin{algorithm}[H]
		\SetKwInOut{Input}{Input}
		\SetKwInOut{Output}{Output}
		\Input{A tuple $(e_{1},e_{2},e_{3})$.}
		\Output{Accept if $(e_{1},e_{2},e_{3})\in R_{\neg e}$, and reject otherwise.}
		\eIf{$e_{1}\geq e_{3}$}{
			\leIf{$e_{2} \leq e_{3}$}{accept} {reject}
		}{
			\eIf{$e_{1}+1=e_{3}$}{
				\leIf{$e_{2}< e_1$}{accept} {reject}
			}{
				\leIf{$e_{2}\leq e_{1}$}{accept} {reject}
			}
		}
		\caption{\label{algo-relation-edge}The procedure to decide $R_{\neg e}$.}
	\end{algorithm}
	
	For the other relations, suppose that $Q$ is a relation on $\mathbf{T}$ which is from the set
	\[\{\mathrm{FUL}_{m},R_{1},\dots,R_{k},(\neg R_{1}),\dots,(\neg R_{k})\},\]
	then $Q^*$ is a relation in the set 
	\[
	\{\mathrm{FUL}_{m}^{*},R_{1}^{*},\dots,R_{k}^{*},(\neg R_{1})^{*},\dots,(\neg R_{k})^{*}\}.
	\]
	If there exists one tuple $(e_{1},e_{1\barwedge2},\dots,e_{n-1\barwedge n},e_{n})\in Q^*$,
	it must be the characteristic tuple 
	\[(d(a_{1}),d(a_{1}\barwedge a_{2}),\dots,d(a_{n-1}\barwedge a_{n}),d(a_{n}))\]
	of some tuple $(a_{1},a_{2}, \dots,a_{n})$ of nodes of $\mathbf{T}$ such that $C^{-1}(\mathbf{T})=\mathbf{A}^+$ and $Qa_{1}a_{2} \dots a_{n}$ holds on $\mathbf{T}$. 
	We divide the process to decide $Q^*$ into several sub-procedures.
	
	Let $D=\{a_{1},\dots,a_{n}\}$ be a set of nodes of $\mathbf{T}$ and $b$ the least common ancestor of the nodes in $D$. Define
	\begin{align*}
		\kappa_{1}= & \min\{d(a_{1}),\dots,d(a_{n})\}
		\\
		\kappa_{2}= & \min\{d(a_{i}\barwedge a_{j})\mid1\leq i<j\leq n\}
		\\
		D_{l}= & \{a\in D\mid a\text{ is in the left  subtree of }b\}
		\\
		D_{r}= & \{a\in D\mid a\text{ is in the right subtree of }b\}
	\end{align*}
	It is easily seen that the following conditions must be satisfied.
	\begin{description}
		\item[(1)] $\kappa_{1}\geq\kappa_{2}$ and $\kappa_{2}=d(b)$.
		\item[(2)] If $\kappa_{1}=\kappa_{2}$, then $b\in D$, $\{b\}\cup D_{l}\cup D_{r}=D$ and
		\begin{itemize}
			\item for any $a_{i}\in D$, $d(b)=d(a_{i})$ implies $b=a_{i}$;
			\item for any $a_{i}\in D_{l}\cup D_{r}$, $d(b\barwedge a_{i})=\kappa_{1}$;
			\item for any $a_{i},a_{j}\in D_{l}$, $d(a_{i}\barwedge a_{j})>\kappa_{1}$;
			\item for any $a_{i},a_{j}\in D_{r}$, $d(a_{i}\barwedge a_{j})>\kappa_{1}$;
			\item for any $a_{i}\in D_{l}$ and $a_{j}\in D_{r}$, $d(a_{i}\barwedge a_{j})=\kappa_{1}$.
		\end{itemize}
		
		\item[(3)] If $\kappa_{1}>\kappa_{2}$, then $b\not\in D$,  $D_{l}\cup D_{r}=D$ and
		
		\begin{itemize}
			\item for any $a_{i},a_{j}\in D_{l}$, $d(a_{i}\barwedge a_{j})>\kappa_{2}$;
			\item for any $a_{i},a_{j}\in D_{r}$, $d(a_{i}\barwedge a_{j})>\kappa_{2}$;
			\item for any $a_{i}\in D_{l}$ and $a_{j}\in D_{r}$, $d(a_{i}\barwedge a_{j})=\kappa_{2}$.
		\end{itemize}
	\end{description}
	
	If $(e_{1},e_{1\barwedge2},\dots,e_{n-1\barwedge n},e_{n})$ is the characteristic tuple of $(a_1,\dots,a_n)$, it must satisfy (1), (2) and (3).
	We present a procedure PreCHECK in Algorithm~\ref{algo-relation-precheck} to check these conditions by dividing a tuple into two parts: the left subtree part $S_{l}$ and right subtree part $S_{r}$. 
	Any legal characteristic tuple can pass the check, but not all tuples accepted by PreCHECK are necessarily characteristic tuples. Note that for the tuple of length 1 or length 3 (e.g., $(e_1)$ or $(e_{1},e_{1\barwedge2},e_{2})$), PreCHCEK can correctly decide whether it is characteristic tuple or not.
	
	\begin{algorithm}[H]
		\SetKwInOut{Input}{Input}
		\SetKwInOut{Output}{Output}
		\Input{A tuple $(e_{1},e_{1\barwedge2},\dots,e_{n-1\barwedge n},e_{n})$.}
		\Output{Accept if the tuple satisfies (1),(2),(3), and reject otherwise.}
		\Switch{$\kappa_{1}\longleftarrow \min\{e_{1},e_{2}\dots,e_{n}\}$, $\kappa_{2} \longleftarrow \min \{e_{i\barwedge j}\mid1\leq i<j\leq n\}$}{
			\lCase{$\kappa_{1}<\kappa_{2}$}{reject}
			\Case{$\kappa_{1}=\kappa_{2}$}{
				$e_{s}\longleftarrow$ the first element in $(e_{1},e_{2}\dots,e_{n})$ that is not equal to $\kappa_{1}$\;
				$S\longleftarrow \{e_{i}\mid1\leq i\leq n\text{ and }e_{i}=\kappa_{1}\}$\;
				$S_{l}\longleftarrow \{e_{i}\mid e_{i}\in\{e_{1},e_{2}\dots,e_{n}\}/S \text{ and }e_{s\barwedge i}>\kappa_{1}\}\cup\{e_{s}\}$\;
				$S_{r}\longleftarrow \{e_{i}\mid e_{i}\in\{e_{1},e_{2}\dots,e_{n}\}/S \text{ and }e_{s\barwedge i}=\kappa_{1}\}$\;
				\leIf{one of the following conditions is satisfied
					\begin{itemize}
						\item $\exists\, e_{i},e_{j}\in S$ such that $e_{i\barwedge j}\neq\kappa_{1}$
						(or $e_{j\barwedge i}\neq\kappa_{1}$);
						\item $\exists\, e_{i}\in S$ and $e_{j}\in S_{l}\cup S_{r}$ such that
						$e_{i\barwedge j}\neq\kappa_{1}$ (or $e_{j\barwedge i}\neq\kappa_{1}$);
						\item $\exists\, e_{i},e_{j}\in S_{l}$ such that $e_{i\barwedge j}=\kappa_{1}$
						(or $e_{j\barwedge i}=\kappa_{1}$);
						\item $\exists\, e_{i},e_{j}\in S_{r}$ such that $e_{i\barwedge j}=\kappa_{1}$
						(or $e_{j\barwedge i}=\kappa_{1}$);
						\item $\exists\, e_{i}\in S_{l}$ and $e_{j}\in S_{r}$ such that $e_{i\barwedge j}\neq\kappa_{1}$(or $e_{j\barwedge i}\neq\kappa_{1}$).
					\end{itemize}
				}{reject}{accept}
			}
			\Case{$\kappa_{1}>\kappa_{2}$}{
				$S_{l}\longleftarrow \{e_{i}\mid1< i\leq n\text{ and }e_{1\barwedge i}>\kappa_{2}\}\cup\{e_{1}\}$\;
				$S_{r}\longleftarrow \{e_{i}\mid1< i\leq n\text{ and }e_{1\barwedge i}=\kappa_{2}\}$\;
				\leIf{one of the following conditions is satisfied
					\begin{itemize}
						\item $\exists\, e_{i},e_{j}\in S_{l}$ such that $e_{i\barwedge j}=\kappa_{2}$
						(or $e_{j\barwedge i}=\kappa_{2}$);
						\item $\exists\, e_{i},e_{j}\in S_{r}$ such that $e_{i\barwedge j}=\kappa_{2}$
						(or $e_{j\barwedge i}=\kappa_{2}$);
						\item $\exists\, e_{i}\in S_{l}$ and $e_{j}\in S_{r}$ such that $e_{i\barwedge j}\neq\kappa_{2}$ (or $e_{j\barwedge i}\neq\kappa_{2}$).
					\end{itemize}
				}{reject}{accept}
			}
		}
		\caption{\label{algo-relation-precheck}The procedure of PreCHECK.}
	\end{algorithm}
	
	The procedure CHECK in Algorithm~\ref{algo-relation-check} decides the tuple whose length is more than 3. CHECK executes by recursively dividing the tuple into the left and right parts, and then invoking itself on these new tuples until the result is obtained.
	
	\begin{algorithm}[H]
		\SetKwInOut{Input}{Input}
		\SetKwInOut{Output}{Output}
		\Input{A tuple $\bar{e}=(e_{1},e_{1\barwedge2},\dots,e_{n-1\barwedge n},e_{n})$.}
		\Output{Accept if $\bar{e}$ is a characteristic tuple, and reject otherwise.}
		\eIf{ $|\bar{e}|=0$ or $|\bar{e}|=1$}{accept\;}{
			\lIf{PreCHECK($\bar{e}$) rejects}{reject}
			\Switch{$\kappa_{1}\longleftarrow \min\{e_{1},\dots,e_{n}\}$,$\kappa_{2}\longleftarrow \min\{e_{i\barwedge j}\mid1\leq i<j\leq n\}$}{
				\Case{$\kappa_{1}=\kappa_{2}$}{
					$e_{s}\longleftarrow$ the first element in $(e_{1},e_{2}\dots,e_{n})$ that is not equal to $\kappa_{1}$\;
					$S\longleftarrow \{e_{i}\mid1\leq i\leq n\text{ and }e_{i}=\kappa_{1}\}$\;
					$S_{l}\longleftarrow \{e_{i}\mid e_{i}\in\{e_{1},\dots,e_{n}\}/S\text{ and }e_{s\barwedge i}>\kappa_{1}\}\cup\{e_{s}\}$\;
					$S_{r}\longleftarrow \{e_{i}\mid e_{i}\in\{e_{1},\dots,e_{n}\}/S\text{ and }e_{s\barwedge i}=\kappa_{1}\}$\;
					$\bar{e}_{l}\longleftarrow (e_{l_{1}},e_{l_{1}\barwedge l_{2}},e_{l_{2}},\dots,e_{l_{h-1}\barwedge l_{h}},e_{l_{h}})$, where $\{e_{l_{1}},\dots,e_{l_{h}}\}=S_{l}$\;
					$\bar{e}_{r}\longleftarrow (e_{r_{1}},e_{r_{1}\barwedge r_{2}},e_{r_{2}},\dots,e_{r_{g-1}\barwedge r_{g}},e_{r_{g}})$, where $\{e_{r_{1}},\dots,e_{r_{g}}\}=S_{r}$\;
					\leIf{CHECK($\bar{e}_{l}$) or CHECK($\bar{e}_{r}$) rejects}{reject}{accept}
				}
				\Case{$\kappa_{1}>\kappa_{2}$}{
					$S_{l}\longleftarrow \{e_{i}\mid1<i\leq n\text{ and }e_{1\barwedge i}>\kappa_{2}\}\cup\{e_{1}\}$\;
					$S_{r}\longleftarrow \{e_{i}\mid1<i\leq n\text{ and }e_{1\barwedge i}=\kappa_{2}\}$\;
					$\bar{e}_{l}\longleftarrow (e_{l_{1}},e_{l_{1}\barwedge l_{2}},e_{l_{2}},\dots,
					e_{l_{h-1}\barwedge l_{h}},e_{l_{h}})$, where $\{e_{l_{1}},\dots,e_{l_{h}}\}=S_{l}$\;
					$\bar{e}_{r}\longleftarrow (e_{r_{1}},e_{r_{1}\barwedge r_{2}},e_{r_{2}},\dots,e_{r_{g-1}\barwedge r_{g}},e_{r_{g}})$, where $\{e_{r_{1}},\dots,e_{r_{g}}\}=S_{r}$\;
					\leIf{CHECK($\bar{e}_{l}$) or CHECK($\bar{e}_{r}$) rejects}{reject}{accept}
				}
			}
		}
		\caption{\label{algo-relation-check}The procedure of CHECK.}
	\end{algorithm}
	
	If $Q^*$ is the relation $\mathrm{FUL}^*$, then for any tuple $\bar{e}$, $\mathrm{FUL}^*\bar{e}$ holds iff $\bar{e}$ passes the procedure CHECK. For the other relations $Q^*\in \{R_{1}^{*},\dots,R_{k}^{*},(\neg R_{1})^{*},\dots,(\neg R_{k})^{*}\}$, we can decide $Q^*$ using Algorithm~\ref{algo-relation-Q}, in which $Q^{-1}$ is the corresponding relation in $\mathbf{A}^+$ that is encoded by $Q$ in $\mathbf{T}$.
	
	\begin{algorithm}[H]
		\SetKwInOut{Input}{Input}
		\SetKwInOut{Output}{Output}
		\Input{A tuple $\bar{e}=(e_{1},e_{1\barwedge2},\dots,e_{r-1\barwedge r},e_{r})$.}
		\Output{Accept if $\bar{e}\in Q^{*}$, and reject otherwise.}
		\eIf{CHECK($\bar{e}$) rejects}{
			reject\;
		}{
			\leIf{$Q^{-1}e_{1} e_{2}\dots e_{r}$ holds in $\mathbf{A}^+$}{accept}{reject}
		}
		\caption{\label{algo-relation-Q}The procedure to decide $Q^*$.}
	\end{algorithm}

	Since the maximal arity of all relations is fixed, both the length of the characteristic tuple and the recursion depth of CHECK are bounded by a constant. 
	It is easy to check that each procedure runs in $O(\log |A|)$ space, where $A$ is the domain of the input structure $\mathbf{A}$.
	Altogether, we can use Algorithm~\ref{algo-relation-euqa}, \ref{algo-relation-edge} and \ref{algo-relation-Q} to enumerate the relations of $S_\mathbf{T}$ in logarithmic space. \qed
\end{proof}

By Lemma~\ref{lem-logspace-comp}, we know that $S_{\mathbf{T}}$ is computable from $\mathbf{A}$ in polynomial time. Hence, $\mathcal{K}$ is in PTIME. 
Since $\mathcal{K}$ is an arbitrary class in EXPTIME, this would imply EXPTIME=PTIME, which contradicts the time hierarchy theorem. So we must have:

\begin{proposition} \label{thm:closureNOTinLFP}
	There is a problem in $\mathrm{PTIME}$ and closed under substructures but not definable in $\mathrm{Datalog}^r$.
\end{proposition}

If a problem is closed under substructures, then its complement is closed under extensions. By Proposition~\ref{prop-datalogrequivlfp}, we can obtain the following corollary.

\begin{corollary}
	$\mathrm{DATALOG}^r\mathrm{[E]} = \mathrm{LFP[E]} \subsetneq \mathrm{PTIME[E]}$.
\end{corollary}

\section{Conclusion}\label{sec-concl}
Revised Datalog is an extension of Datalog by allowing universal quantification over intensional relations in the body of rules. On all finite structures, Datalog$^r$ is strictly more expressive than Datalog, and has the same expressive power as that of LFP. In classical model theory, the closure properties of a formula are usually related to some syntactic properties. Many preservation theorems are proved to reflect this relationship. When restricted to finite structures, most of these preservation theorems fail. From the syntax and semantics of Datalog, we can treat it as the dual of SO-HORN logic, which is closed under substructures~\cite{gradel1991expressive}. It follows that Datalog is closed under extensions. A lot of work has been conducted between Datalog and FO (or LFP) to study the closure property. In this paper, we study the expressive power of revised Datalog on the problems that are closed under substructures.
We show that Datalog$^r$ cannot define all the problems that are in PTIME and closed under substructures. As a corollary, LFP cannot define all the extension-closed problems that are in PTIME. A method of tree encodings for arbitrary structures is used in the proof. If we replace the extension closure property by the homomorphism closure property, it is still open whether the statement also holds. This desirable for future work.

\bibliographystyle{splncs04}
\bibliography{./reference}

\end{document}